\documentclass[12pt,journal,onecolumn]{IEEEtran}
\usepackage{graphicx,psfrag,epsfig,epsf,latexsym,hhline,amsmath,amssymb,multirow,pstricks}
\usepackage{comment}
\usepackage{hyperref}
\IEEEoverridecommandlockouts
\newtheorem{theorem}{Theorem}
\newtheorem{lemma}[theorem]{Lemma}

\newtheorem{definition}[theorem]{Definition}
\newtheorem{corollary}[theorem]{Corollary}
\begin{document}
\title{Lattices over Eisenstein Integers for Compute-and-Forward}
\author{Nihat Engin Tunali\dag, Yu-Chih Huang\dag, Joseph J. Boutros\S, and\\
Krishna R. Narayanan\dag\\
\dag Department of Electrical and Computer Engineering, Texas A\&M University\\
\S Department of Electrical Engineering, Texas A\&M University at Qatar\\
{\tt\small {\{engintunali@gmail.com, jerry.yc.huang@gmail.com, boutros@tamu.edu, krn@ece.tamu.edu\}} }
\thanks{The paper was presented in part at the 2011 Banff Workshop on Algebraic Structure in Network Information Theory,
at the 2012 Information Theory and its Application Workshop, and the 2012 Allerton Conference on Communications, Control and
Computing.}}

\maketitle
\begin{abstract}
In this paper, we consider the use of lattice codes over Eisenstein integers for implementing a compute-and-forward protocol
in wireless networks when channel state information is not available at the transmitter. We extend the compute-and-forward
paradigm of Nazer and Gastpar to decoding Eisenstein integer combinations of transmitted messages at relays by proving the existence
of a sequence of pairs of nested lattices over Eisenstein integers in which the coarse lattice is good for covering and the fine lattice
can achieve the Poltyrev limit. Using this result, we show that both the outage performance and error-correcting performance of nested lattice codebooks over Eisenstein integers surpasses lattice codebooks over integers considered by Nazer and Gastpar with no additional computational complexity.

\end{abstract}


\section{Introduction}
Compute-and-forward is a novel relaying paradigm in wireless communications in which relays in a network directly compute or decode functions of signals
transmitted from multiple transmitters and forward them to a central destination. One of the most effective ways to implement a compute-and-forward scheme
is to employ lattice codes at each transmitter. Since a lattice is closed under integer addition, lattice codes are naturally suited to decoding integer
linear combinations of transmitted signals.

Lattice codes have been shown to be optimal for several problems in communications including coding for the point-to-point additive white Gaussian noise
(AWGN) channel \cite{erez2004achieving} and coding with side information problems such as
the dirty paper coding problem and Wyner-Ziv problem \cite{zamirmultiterminal}. The construction of optimal lattice codes for these problems
requires a lattice that is good for channel coding. Since a lattice has unconstrained power, goodness for channel coding is measured using Poltyrev's idea of the unconstrained AWGN channel. In \cite{Polty}, Poltyrev derives the maximum noise variance that a lattice can tolerate while maintaining reliable communication over the unconstrained point-to-point AWGN channel, which is referred to as the Poltyrev limit in literature.
Loeliger showed the existence of lattices that achieve the Poltyrev limit by means of Construction A in \cite{loeliger}.
Then, Erez \textit{et} \textit{al.}, showed that there exists lattices which are simultaneously good for quantization and can achieve the Poltyrev limit in \cite{erez2005lattices} which made it possible to construct nested lattice
codes that were able to achieve a rate of $\frac{1}{2}\log{\left(1+\text{SNR}\right)}$ over the point-to-point AWGN channel. There has also been great interest in constructing lattice codes with reasonable encoding and decoding complexities such as Signal Codes and Low Density Lattice Codes \cite{SC}, \cite{LDLC}.

In a bidirectional relay network when channel state information is available at the transmitters, the transmitters can compensate for the channel gains and the relay can decode to the sum of the transmitted signals, which is a special case of compute-and-forward. For this system model, it was shown that an exchange rate of $\frac{1}{2}\log{\left(\frac{1}{2}+\text{SNR}\right)}$ can be achieved using nested lattice codes at the transmitters, which is optimal for asymptotically large signal-to-noise ratios and provides substantial gains over other relaying paradigms such as amplify-and-forward and decode-and-forward \cite{wilson}, \cite{Nam}. In \cite{Niesen}, a novel compute-and-forward implementation is proposed for the $K\times K$ AWGN interference network where channel state information is available at the transmitters, which achieves the full $K$ degrees of freedom.

We consider the case when channel state information is not available at the transmitters. In this case, an effective way to implement a compute-and-forward scheme is to allow the relay to adaptively choose the integer coefficients depending on the channel coefficients. Nazer and Gastpar have introduced and analyzed such a scheme which uses lattices over integers and they have derived achievable information rates in \cite{nazer2011CF}. In \cite{feng}, Feng, Silva and Kschischang have introduced an algebraic framework for designing lattice codes for compute-and-forward. The framework in \cite{feng} is quite general in the sense that every lattice partition based compute-and-forward scheme can be put into this framework, including the one by Nazer and Gastpar in \cite{nazer2011CF}. However, \cite{feng} does not provide a means to identify good lattice partition based schemes.

In this paper, we contribute to the literature by identifying a lattice partition based compute-and-forward scheme which is particularly good for approximating channel coefficients from the complex field. Our scheme can be regarded as an extension of the scheme in \cite{nazer2011CF} to lattices over Eisenstein integers. We show that an improvement in outage performance and error-correcting performance can be obtained compared to using lattices over integers. We proceed by proving the existence of a sequence of nested lattices over Eisenstein integers in which the coarse lattice is good for covering and the fine lattice achieves the Poltyrev limit. Using this result, we can show similar results to those in \cite{nazer2011CF} with integers replaced by Eisenstein integers. The main improvement in outage and error-correcting performance is a consequence of that the use of lattices over Eisenstein integers permits the relay to decode to a linear combination of the transmitted signals where the coefficients are Eisenstein integers, which quantize channel coefficients better than Gaussian integers.

\color{black}Recently, we became aware of an independent work by Sun \emph{et. al.} \cite{other_eis} where lattice network codes over Eisenstein integers are also considered. The main focus in \cite{other_eis} is the analysis of the decoding error probability, which suggests that lattice network codes built over Eisenstein integers can provide significant coding gains over lattice network codes built over Gaussian integers. Our work differs from \cite{other_eis} in the following ways. While their focus is on constructing finite constellations from lattice partitions which are suitable for compute-and-forward, we consider construction of lattices (infinite constellations) over Eisenstein integers and show the optimality of such construction. Moreover, their coding scheme can be regarded as the concatenation of a linear code over an appropriate finite field and a constellation carved from a lattice partition. On the other hand, our scheme is a more general one which is formed by the quotient group of a lattice over Eisenstein integers and its sublattice. It can be shown that the scheme in \cite{other_eis} is a special case of ours with hypercube shaping\footnote{Here, we use the term ``hypercube shaping" to denote a scheme using a properly scaled version of Eisenstein integers as shaping (coarse) lattice. Thus, when $\mathbb{Z}$ or $\mathbb{Z}[i]$ are considered, the shape is a hypercube. However, it is in fact not a hypercube if $\mathbb{Z}[\omega]$ is considered.}. This generalization is imperative in the sense that it allows us to show the achievable computation rates if one would use such lattices for compute-and-forward.   \color{black}

The structure of our paper is as follows. In Section \ref{Sec:Notational Convention}, we introduce the notation that will be used throughout the paper.
In Section \ref{Sec: System Model}, we present the system model that will be considered. In Section \ref{Sec: Background on lattices},
we provide some background on lattices and lattice codes. In Section \ref{cfgausinf}, we discuss Nazer and Gastpar's framework for compute-and-forward
\cite{nazer2011CF}. In Section~\ref{CFeisinf}, we discuss how lattices over Eisenstein integers can be used for compute-and-forward in Nazer and Gastpar's
framework and what properties of these lattices are required in order to achieve computation rates formulated similarly to those in \cite{nazer2011CF}. In Section~\ref{Sec: num_res}, we provide numerical results and compare the outage performance and error-correcting performance of lattices over natural integers and lattices over Eisenstein integers in compute-and-forward. In Appendix~\ref{Sec: Appendix Notation}, we introduce the notation that is
used in Appendix~\ref{sec:GfC} and Appendix~\ref{sec:GFACC}, we prove that there exist a nested pair of Eisenstein lattices
which the coarse lattice is good for covering and the fine lattice achieves the Poltyrev limit.

\subsection{Notational Convention}\label{Sec:Notational Convention}
Throughout the paper, we use $\mathbb{R}$ to denote the field of real numbers, $\mathbb{C}$ to denote the field of complex numbers,
 and $\mathbb{F}_q$ to denote a finite field of size $q$. $\mathbb{Z}$, $\mathbb{Z}[i]$, and $\mathbb{Z}[\omega]$ are used to denote
 the set of integers, Gaussian integers, and Eisenstein integers, respectively. We use underlined variables to denote vectors and
 boldface uppercase variables to denote matrices, e.g., $\underline{x}$ and $\mathbf{X}$, respectively. We denote the $i^{th}$ column
of a matrix $\mathbf{X}$ as $\mathbf{X}_i$. Also, we use superscript $H$ to denote the Hermitian operation, e.g., $\underline{x}^H$
and $\mathbf{X}^H$. We define $\log^{+}(x)\triangleq\max(\log_2(x),0)$ and denote the Euclidean  metric as $\|\cdot\|$. We denote the all zero vector in $\mathbb{R}^n$ as $\underline{0}$ and the $n\times n$ identity matrix as $\mathbf{I}$. We denote the volume of a bounded region $E\subset\mathbb{R}^n$ as $\text{Vol}\left(E\right)$ and denote the $n$-dimensional sphere of radius $r$ centered at $\underline{0}$ as $\mathcal{B}(r)\triangleq\left\{\underline{s}~:~\|\underline{s}\|\leq r \right\}$.

\section{System Model}\label{Sec: System Model}
We consider an AWGN network as shown in Fig.~\ref{fig-AWGN Network} where $L$ source nodes $S_1, S_2,\ldots, S_L$ wish to transmit information to $M$ relay nodes $D_1, D_2,\ldots, D_M$, where $M\geq L$. It is assumed that relay nodes cannot collaborate with each other and are noiselessly connected to a final destination interested in the individual messages sent from all the source nodes. The objective of the relay nodes is to facilitate communication between the source nodes and the final destination.

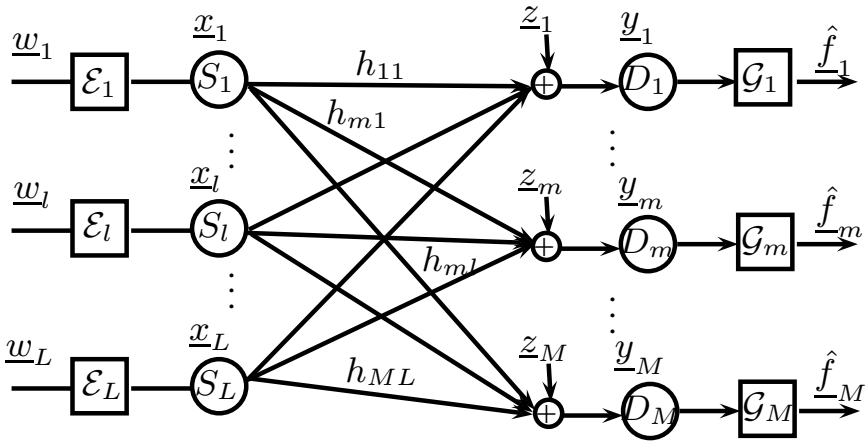
\begin{figure}[h]
\centering
\scalebox{1.5} 
{
\begin{pspicture}(0,-1.9348438)(9.887188,1.9748437)
\usefont{T1}{ptm}{m}{n}
\rput(7.6867185,-1.6398437){\footnotesize $\mathcal{G}_M$}
\psline[linewidth=0.04cm,arrowsize=0.05291667cm 2.0,arrowlength=1.4,arrowinset=0.4]{->}(6.8887496,-1.6748438)(7.4287496,-1.6748438)
\psframe[linewidth=0.04,dimen=outer](7.9287496,-1.3948437)(7.4087496,-1.9148438)
\psline[linewidth=0.04cm,arrowsize=0.05291667cm 2.0,arrowlength=1.4,arrowinset=0.4]{->}(7.94875,-1.6748438)(8.4887495,-1.6748438)
\psframe[linewidth=0.04,dimen=outer](7.9087496,0.08515645)(7.3887496,-0.43484354)
\usefont{T1}{ptm}{m}{n}
\rput(7.6567187,-0.15984355){\footnotesize $\mathcal{G}_m$}
\psline[linewidth=0.04cm,arrowsize=0.05291667cm 2.0,arrowlength=1.4,arrowinset=0.4]{->}(7.9287496,-0.19484355)(8.46875,-0.19484355)
\psline[linewidth=0.04cm,arrowsize=0.05291667cm 2.0,arrowlength=1.4,arrowinset=0.4]{->}(6.8687496,-0.19484355)(7.4087496,-0.19484355)
\usefont{T1}{ptm}{m}{n}
\rput(7.6067185,1.2801565){\footnotesize $\mathcal{G}_1$}
\psline[linewidth=0.04cm,arrowsize=0.05291667cm 2.0,arrowlength=1.4,arrowinset=0.4]{->}(7.9287496,1.2451564)(8.46875,1.2451564)
\psline[linewidth=0.04cm,arrowsize=0.05291667cm 2.0,arrowlength=1.4,arrowinset=0.4]{->}(6.8687496,1.2451564)(7.4087496,1.2451564)
\psline[linewidth=0.04cm,arrowsize=0.05291667cm 2.0,arrowlength=1.4,arrowinset=0.4]{->}(3.0487502,1.2251563)(5.5687494,1.2051563)
\usefont{T1}{ptm}{m}{n}
\rput(5.666719,0.36015648){\footnotesize $\underline{z}_m$}
\usefont{T1}{ptm}{m}{n}
\rput(5.696719,-1.1198436){\footnotesize $\underline{z}_M$}
\pscircle[linewidth=0.04,dimen=outer](5.7087493,1.2251563){0.14}
\usefont{T1}{ptm}{m}{n}
\rput(5.689843,1.2201563){\footnotesize +}
\psline[linewidth=0.04cm,arrowsize=0.05291667cm 2.0,arrowlength=1.4,arrowinset=0.4]{->}(5.828749,1.2051563)(6.3687496,1.2051563)
\psline[linewidth=0.04cm,arrowsize=0.05291667cm 2.0,arrowlength=1.4,arrowinset=0.4]{->}(5.7087493,1.6651562)(5.7287493,1.3251563)
\usefont{T1}{ptm}{m}{n}
\rput(5.6367188,1.8001562){\footnotesize $\underline{z}_1$}
\pscircle[linewidth=0.04,dimen=outer](5.7087493,-0.21484356){0.14}
\usefont{T1}{ptm}{m}{n}
\rput(5.689843,-0.21984357){\footnotesize +}
\psline[linewidth=0.04cm,arrowsize=0.05291667cm 2.0,arrowlength=1.4,arrowinset=0.4]{->}(5.828749,-0.23484357)(6.3687496,-0.23484357)
\psline[linewidth=0.04cm,arrowsize=0.05291667cm 2.0,arrowlength=1.4,arrowinset=0.4]{->}(5.7087493,0.22515646)(5.7287493,-0.11484355)
\usefont{T1}{ptm}{m}{n}
\rput(6.296719,-0.7198435){\footnotesize $\vdots$}
\pscircle[linewidth=0.04,dimen=outer](5.7287493,-1.6948433){0.14}
\usefont{T1}{ptm}{m}{n}
\rput(5.709843,-1.6998434){\footnotesize +}
\psline[linewidth=0.04cm,arrowsize=0.05291667cm 2.0,arrowlength=1.4,arrowinset=0.4]{->}(5.848749,-1.7148434)(6.3887496,-1.7148434)
\psline[linewidth=0.04cm,arrowsize=0.05291667cm 2.0,arrowlength=1.4,arrowinset=0.4]{->}(5.7287493,-1.2548436)(5.7487497,-1.5948435)
\usefont{T1}{ptm}{m}{n}
\rput(6.5267186,1.7201563){\footnotesize $\underline{y}_1$}
\usefont{T1}{ptm}{m}{n}
\rput(6.556719,0.26015642){\footnotesize $\underline{y}_m$}
\usefont{T1}{ptm}{m}{n}
\rput(6.5467186,-1.1998434){\footnotesize $\underline{y}_M$}
\usefont{T1}{ptm}{m}{n}
\rput(2.6967187,0.3801564){\footnotesize $\underline{x}_l$}
\usefont{T1}{ptm}{m}{n}
\rput(2.7367187,-1.0198437){\footnotesize $\underline{x}_L$}
\usefont{T1}{ptm}{m}{n}
\rput(4.246719,1.4001565){\footnotesize $h_{11}$}
\usefont{T1}{ptm}{m}{n}
\rput(4.866719,-0.33984354){\footnotesize $h_{ml}$}
\usefont{T1}{ptm}{m}{n}
\rput(4.2567186,-1.3198438){\footnotesize $h_{ML}$}
\usefont{T1}{ptm}{m}{n}
\rput(4.0367184,0.9601563){\footnotesize $h_{m1}$}
\usefont{T1}{ptm}{m}{n}
\rput(1.7767185,-1.4398438){\footnotesize $\mathcal{E}_L$}
\psframe[linewidth=0.04,dimen=outer](2.00875,-1.1948438)(1.4887499,-1.7148434)
\psline[linewidth=0.04cm](2.0287497,-1.4748435)(2.56875,-1.4748435)
\psline[linewidth=0.04cm](0.96874994,-1.4748435)(1.50875,-1.4748435)
\psframe[linewidth=0.04,dimen=outer](2.00875,0.20515648)(1.4887499,-0.31484354)
\usefont{T1}{ptm}{m}{n}
\rput(1.7367185,-0.03984355){\footnotesize $\mathcal{E}_l$}
\psline[linewidth=0.04cm](2.0287497,-0.07484355)(2.56875,-0.07484355)
\psline[linewidth=0.04cm](0.96874994,-0.07484355)(1.50875,-0.07484355)
\psline[linewidth=0.04cm](0.96874994,1.2451564)(1.50875,1.2451564)
\usefont{T1}{ptm}{m}{n}
\rput(1.7467185,1.2401565){\footnotesize $\mathcal{E}_1$}
\psline[linewidth=0.04cm](2.0287497,1.2451564)(2.56875,1.2451564)
\usefont{T1}{ptm}{m}{n}
\rput(2.7467186,1.7201563){\footnotesize $\underline{x}_1$}
\usefont{T1}{ptm}{m}{n}
\rput(1.1667186,1.5601563){\footnotesize $\underline{w}_1$}
\pscircle[linewidth=0.04,dimen=outer](2.8087502,1.2651565){0.26}
\usefont{T1}{ptm}{m}{n}
\rput(2.7767186,1.2801565){\footnotesize $S_1$}
\pscircle[linewidth=0.04,dimen=outer](6.6087494,1.2451564){0.26}
\usefont{T1}{ptm}{m}{n}
\rput(6.576719,1.2601564){\footnotesize $D_1$}
\pscircle[linewidth=0.04,dimen=outer](2.8087502,-0.05484355){0.26}
\usefont{T1}{ptm}{m}{n}
\rput(2.7675936,-0.03984355){\footnotesize $S_l$}
\pscircle[linewidth=0.04,dimen=outer](6.6087494,-0.19484355){0.26}
\usefont{T1}{ptm}{m}{n}
\rput(6.606719,-0.17984354){\footnotesize $D_m$}
\usefont{T1}{ptm}{m}{n}
\rput(2.9167187,-0.49984354){\footnotesize $\vdots$}
\pscircle[linewidth=0.04,dimen=outer](2.8087502,-1.4548438){0.26}
\usefont{T1}{ptm}{m}{n}
\rput(2.7867186,-1.4398438){\footnotesize $S_L$}
\pscircle[linewidth=0.04,dimen=outer](6.6287494,-1.6748438){0.26}
\usefont{T1}{ptm}{m}{n}
\rput(6.6367188,-1.6598438){\footnotesize $D_M$}
\usefont{T1}{ptm}{m}{n}
\rput(8.266719,1.5201565){\footnotesize $\underline{\hat{f}}_1$}
\psline[linewidth=0.04cm,arrowsize=0.05291667cm 2.0,arrowlength=1.4,arrowinset=0.4]{->}(3.08875,1.1851563)(5.609218,-1.665156)
\psline[linewidth=0.04cm,arrowsize=0.05291667cm 2.0,arrowlength=1.4,arrowinset=0.4]{->}(3.0728123,-0.06515615)(5.6087494,1.2051563)
\psline[linewidth=0.04cm,arrowsize=0.05291667cm 2.0,arrowlength=1.4,arrowinset=0.4]{->}(3.0928123,-0.10515615)(5.589218,-1.6851561)
\usefont{T1}{ptm}{m}{n}
\rput(6.296719,0.76015645){\footnotesize $\vdots$}
\usefont{T1}{ptm}{m}{n}
\rput(2.9167187,0.7401565){\footnotesize $\vdots$}
\psline[linewidth=0.04cm,arrowsize=0.05291667cm 2.0,arrowlength=1.4,arrowinset=0.4]{->}(3.052812,-1.4051563)(5.5687494,1.1851563)
\psline[linewidth=0.04cm,arrowsize=0.05291667cm 2.0,arrowlength=1.4,arrowinset=0.4]{->}(3.0728123,-1.3851563)(5.589218,-1.7051562)
\psline[linewidth=0.04cm,arrowsize=0.05291667cm 2.0,arrowlength=1.4,arrowinset=0.4]{->}(3.06875,1.2051563)(5.612812,-0.18515617)
\psline[linewidth=0.04cm,arrowsize=0.05291667cm 2.0,arrowlength=1.4,arrowinset=0.4]{->}(3.0487502,-0.09484355)(5.592812,-0.18515617)
\psline[linewidth=0.04cm,arrowsize=0.05291667cm 2.0,arrowlength=1.4,arrowinset=0.4]{->}(3.06875,-1.3948437)(5.612812,-0.18515617)
\usefont{T1}{ptm}{m}{n}
\rput(8.316719,0.1001564){\footnotesize $\underline{\hat{f}}_m$}
\usefont{T1}{ptm}{m}{n}
\rput(8.326719,-1.3598435){\footnotesize $\underline{\hat{f}}_M$}
\usefont{T1}{ptm}{m}{n}
\rput(1.1567186,0.2401564){\footnotesize $\underline{w}_l$}
\usefont{T1}{ptm}{m}{n}
\rput(1.1367188,-1.1598437){\footnotesize $\underline{w}_L$}
\psframe[linewidth=0.04,dimen=outer](2.00875,1.5051564)(1.4887499,0.9851565)
\psframe[linewidth=0.04,dimen=outer](7.8887496,1.5451565)(7.3687496,1.0251565)
\end{pspicture}
}\caption{The AWGN Network where $S_1, S_2,\ldots, S_L$ wish to transmit information to $D_1, D_2,\ldots, D_M$. The channel between the $S_l$ and $D_m$ is denoted as $h_{ml}$.}\label{fig-AWGN Network}
\end{figure}

We denote the information vector at the source node $S_l$ as $\underline{w}_l\in\mathbb{F}_q^{k}$. Without loss of generality, we assume that the length of the information vector at each transmitter $l$ has the same length $k$. Each transmitter is equipped with an encoder $\mathcal{E}_l:\mathbb{F}_q^k\rightarrow\mathbb{C}^n$ that maps $\underline{w}_l$ to an $n$-dimensional complex codeword $\underline{x}_l=\mathcal{E}_l\left(\underline{w}_l\right)$. Each codeword is subject to the power constraint

\begin{align}
    \mathbb{E}||\underline{x}_l||^2\leq nP.
\end{align}
The message rate $R$ of each transmitter is the length of its message in bits normalized by the number of channel uses,
\begin{align}
R=\frac{k}{n}\log{q}.
\end{align}
Due to the superposition nature of the wireless medium, each relay $m$ observes
\begin{align}\label{eq:relay_obsv}
    \underline{y}_m=\sum_{l=1}^L h_{ml}\underline{x}_l+\underline{z}_m,
\end{align}
where $h_{ml}\in\mathbb{C}$ is the channel coefficient between $D_m$ and $S_l$. \color{black} As it can be observed from \eqref{eq:relay_obsv}, it is assumed that there is no inter-symbol interference and each $h_{ml}\underline{x}_l$ arrive at the relay simultaneously. \color{black} Furthermore, $\underline{z}_m$ is an $n$-dimensional complex vector which consists of identically distributed (i.i.d.) circularly symmetric Gaussian random variables, i.e. $\underline{z}_m\sim \mathcal{CN}(0,\textbf{I})$. Let $\underline{h}_m=[h_{m1},\cdots,h_{mL}]^T$ denote the vector of channel coefficients to relay $m$ from all the source nodes. We assume that the relay $m$ only has the knowledge of the channel coefficient from each transmitter to itself, i.e., $\underline{h}_m$.


Each relay attempts to recover the linear combination $\underline{f}_m$ (over $\mathbb{F}_q$)
\begin{align}
    \underline{f}_m=\bigoplus_{l=1}^L\left( b_{ml}\underline{w}_l \right),
\end{align}
where $b_{ml}\in \mathbb{F}_q$ and let $\underline{b}_m=[b_{m1},\ldots,b_{mL}]^T$. Typically $b_{ml}$s are chosen based on the network structure and/or the channel coefficients.
It is desirable for the matrix $[\underline{b}_1,\ldots,\underline{b}_{M}]$ to be full-rank which enables each $\underline{w}_l$ to be recovered at the final destination. For each $D_m$, we define the decoder $\mathcal{G}_m:\mathbb{C}^n\rightarrow\mathbb{F}_q^k$ and $\hat{\underline{f}}_m=\mathcal{G}_m(y_m)$ is an estimate of $\underline{f}_m$. \color{black} Let $\mathcal{P}$ denote a principal ideal domain in $\mathbb{C}$ such as $\mathbb{Z}[i]$ or $\mathbb{Z}[\omega]$. \color{black}

\begin{definition}[Average probability of error]\label{averr}
Equations with coefficient vectors $\underline{a}_1,\underline{a}_2,\ldots\underline{a}_M$, where each $\underline{a}_m\in\mathcal{P}^L$, are decoded with \textit{average probability of error} $\epsilon$ if

\begin{align}
    \text{Pr}\left ( \bigcup_{m=1}^M\left\lbrace \hat{\underline{f}}_m\neq\underline{f}_m\right\rbrace \right )<\epsilon.
\end{align}

\end{definition}

\begin{definition}[Computation rate of relay $m$]\label{achcomprate}
For a given channel coefficient vector $\underline{h}_m$ and equation coefficient vector $\underline{a}_m\in\mathcal{P}^L$, the computation rate $R\left(\underline{h}_m,\underline{a}_m \right)$ is achievable at relay $m$ if for any $\epsilon>0$ and $n$ large enough, there exist encoders $\mathcal{E}_1,\ldots,\mathcal{E}_L$ and there exists a decoder $\mathcal{G}_m$ such that relay $m$ can recover its desired equation with average probability of error $\epsilon$ as long as the underlying message rate $R$ satisfies
\begin{align}
    R<R\left(\underline{h}_m,\underline{a}_m \right).
\end{align}



\end{definition}
Due to the fact that the relays cannot collaborate, each relay picks an integer vector $\underline{a}_m$ such that $R\left(\underline{h}_m,\underline{a}_m \right)$ is maximized.

\begin{definition}[Computation rate of AWGN network]\label{achcompratenet}
Given  $\mathbf{H}=[\underline{h}_1,\ldots,\underline{h}_M]$ and $\mathbf{A}=[\underline{a}_1,\ldots,\underline{a}_M]$, the achievable computation rate of an AWGN network is defined as
\begin{align}
    \mathcal{R}\left(\mathbf{H},\mathbf{A} \right)=\underset{m:a_{ml}\neq0}{\min}R\left(\underline{h}_m,\underline{a}_m \right),
\end{align}
\color{black}provided that the matrix $\sigma\left(\mathbf{A}\right)=[\underline{b}_1,\ldots,\underline{b}_{M}]\in\mathbb{F}_q^{L\times M}$, where $\sigma:\mathcal{P}^{L\times M}\rightarrow\mathbb{F}_q^{L\times M}$, is full rank. \color{black} If $[\underline{b}_1,\ldots,\underline{b}_{M}]$ is not full rank, $\mathcal{R}\left(\mathbf{H},\mathbf{A} \right)=0$.

%
\end{definition}
{\black Note that in this paper, our coding scheme particular considers the ring of Eisenstein integers, i.e., $\mathcal{P}=\mathbb{Z}[\omega]$, for the reason that will become clear later.}

\section{Background on Lattices}\label{Sec: Background on lattices}
Due to the fact that the coding scheme that will be considered heavily relies on lattices, we now provide some background knowledge
on lattices. For more details on lattices, please refer to \cite{conway1999sphere}, \cite{erez2005lattices}, and \cite{erez2004achieving}.

\begin{definition}[Lattice over $\mathbb{Z}$]
An $n$-dimensional \textit{lattice over natural integers}, $\Lambda^{(n)}$, is a discrete set of points in $\mathbb{R}^n$ such that $\Lambda^{(n)}$ is a discrete additive subgroup of $\mathbb{R}^n$ with rank $k$ where $k\leq n$. Such a lattice can be generated via a full rank generator matrix $\mathbf{B}\in\mathbb{R}^{n\times k}$
\begin{align}
    \Lambda^{(n)}=\left\{\underline{\lambda}=\mathbf{B}\underline{e}:\underline{e}\in\mathbb{Z}^k \right\}.
\end{align}
For notational convenience, we shall drop the superscript in $\Lambda^{(n)}$ in this paper and denote
$n$-dimensional lattices as $\Lambda$. Also, we refer to lattices over integers as $\mathbb{Z}$-lattices throughout the paper.
\end{definition}

\color{black}Given a lattice $\Lambda$, we denote the \emph{quantizer} operation with respect to $\Lambda$ as $Q_{\Lambda}$, the \emph{modulus} operation with respect to $\Lambda$ as $\mod\Lambda$, and the \emph{fundamental Voronoi region} of $\Lambda$ as $\mathcal{V}_\Lambda$. We denote the \emph{covering radius} and \emph{effective radius} of $\Lambda$ as $r_{\Lambda}^{\text{cov}}$ and $r_{\Lambda}^{\text{eff}}$, respectively. We denote the \emph{second moment} and \emph{normalized second moment} of $\Lambda$ as $\sigma_\Lambda^2$ and $G\left(\Lambda\right)$, respectively. We refer the readers to \cite{conway1999sphere} for these definitions.\color{black}

 \begin{definition}[Goodness for covering]
 A sequence of lattices $\Lambda$ is \emph{good for covering} if
 \begin{align}
    \underset{n\rightarrow\infty}{\lim}\frac{r_{\Lambda}^{{\text{cov}}}}{r_{\Lambda}^{{\text{eff}}}}=1.
 \end{align}
 These lattices are also commonly referred to as \emph{Rogers good}, since it was first shown by Rogers that such lattices exist \cite{Rogers}.
 \end{definition}

 \begin{definition}[Goodness for quantization]
 A sequence of lattices $\Lambda$ is \emph{good for quantization} if
 \begin{align}
 \underset{n\rightarrow\infty}{\lim}G\left(\Lambda\right)=\frac{1}{2\pi e}.
 \end{align}
 In other words, the normalized second moment of $\Lambda$ converges to a sphere's normalized second moment as $n\rightarrow\infty$. Zamir \emph{et al.}, have shown that such a sequence of lattices exist \cite{Zamir-Feder1996Quan}. Erez \emph{et al.} have also shown the existence of such a sequence of lattices and proved that goodness for covering implies goodness for quantization \cite{erez2005lattices}.
 \end{definition}

 \begin{definition}[Lattices that achieve the Poltyrev limit]
    Let $\underline{z}$ be an $n$-dimensional independent and identically distributed (i.i.d) Gaussian vector, $\underline{z}\sim\mathcal{N}\left( \underline{0},\theta_{{\underline{z}}}^2\mathbf{I}\right)$. The \emph{effective radius} of $\underline{z}$, which we denote as $r_{\underline{z}}$, is defined as
    \begin{align}
    r_{\underline{z}}=\sqrt{n\theta_{{\underline{z}}}^2}.
    \end{align}
   Consider a $\mathbb{Z}$-lattice $\Lambda$ and a lattice point $\underline{\lambda}\in\Lambda$, which is transmitted across an AWGN channel:
   \begin{align}
   \underline{y}=\underline{\lambda}+\underline{z}.
   \end{align}
    The maximum likelihood decoder would decode to the lattice point nearest in Euclidean distance to $\underline{y}$. Therefore, an error would occur only if $\underline{y}$ leaves the Voronoi region of $\underline{\lambda}$. Due to lattice symmetry, this is equivalent to $\underline{z}$ leaving the fundamental Voronoi region $\mathcal{V}_\Lambda$.
    \begin{align}
    P_e\left(\Lambda,r_{\underline{z}}\right)=\text{Pr}\left\{ \underline{z}\not\in\mathcal{V}_{\Lambda}\right\},
    \end{align}
    where $P_e\left(\Lambda,r_{\underline{z}}\right)$ denotes the probability of error.

    A sequence of $\mathbb{Z}$-lattices $\Lambda$ are \emph{good for AWGN channel coding} if for any $r_{\underline{z}}<r_{\Lambda}^{\text{eff}}$, $\underset{n\rightarrow\infty}\lim{P_e}\left(\Lambda,r_{\underline{z}}\right)=0$ and this decay may be bounded exponentially in $n$. Erez \textit{et. al.} have shown the existence of such a sequence of lattices in \cite{erez2005lattices} and they have referred to them as \emph{Poltyrev good}.

    Nonetheless, in order to achieve the Poltyrev capacity in the unconstrained AWGN channel, it is sufficient for $\lim\underset{n\rightarrow\infty}{P_e}\left(\Lambda,r_{\underline{z}}\right)=0$ for any $r_{\underline{z}}<r_{\Lambda}^{\text{eff}}$, i.e., $P_e\left(\Lambda,r_{\underline{z}}\right)$ does not need
to decay exponentially as $n\rightarrow\infty$. We refer to such a sequence of lattices as \emph{lattices that achieve the Poltyrev limit} in this paper. Loeliger has shown the existence of such lattices in \cite{loeliger}.

 \end{definition}


 \begin{definition}[Sublattice]
 A $\mathbb{Z}$-lattice $\Lambda$ is a sublattice of (nested in) another $\mathbb{Z}$-lattice $\Lambda_f$ if $\Lambda\subseteq\Lambda_f$. $\Lambda$ is referred to as the \emph{coarse lattice} and $\Lambda_f$ is referred to as the \emph{fine lattice}. The quotient group $\Lambda_f/\Lambda$ is referred to as a lattice partition \cite{forney}.
 \end{definition}

\begin{definition}[Nesting ratio]
Given a pair of $n$-dimensional nested lattices $\Lambda\subset\Lambda_f$, the \textit{nesting ratio} $\vartheta$ is defined as,
\begin{align}
\vartheta=\left(\frac{\text{Vol}(\mathcal{V}_\Lambda)}{\text{Vol}(\mathcal{V}_{\Lambda_f})}\right)^{\frac{1}{n}}.
\end{align}

\end{definition}

 \begin{definition}[Nested Lattice Code]
 Given a fine $\mathbb{Z}$-lattice $\Lambda_f$ and a coarse $\mathbb{Z}$-lattice $\Lambda$, where $\Lambda\subseteq\Lambda_f$,  a \emph{nested lattice code} (Voronoi code), which we refer to as $\mathcal{L}$, is the set of all coset leaders in $\Lambda_f$ that lie in the fundamental Voronoi region of the coarse lattice $\Lambda$ \cite{Voronoi_Codes}:
 \begin{align}
 \mathcal{L}=\mathcal{V}_\Lambda\cap\Lambda_f=\left\{\underline{\lambda}_f~:~Q_{\Lambda}\left({\underline{\lambda}_f}\right)=\underline{0},\underline{\lambda}_f\in\Lambda_f \right\}.
 \end{align}
 In other words, $\mathcal{L}$ is a set of coset representatives of the quotient group $\Lambda_f/\Lambda$.

The \textit{coding rate} of a nested lattice code, denoted as $R$ is defined as,
\begin{align}
R=\log{\vartheta}.
\end{align}
\end{definition}

    \subsection{Construction A for $\mathbb{Z}$-lattices}

    One way to construct $\mathbb{Z}$-lattices is to use the following procedure, which is referred to as \emph{Construction A} \cite{Constr_A_Paper}:

    Let $q$ be a natural prime and $k,n$ be integers such that $k\leq n$. Then, let $\mathbf{{G}}\in\mathbb{F}_q^{n\times k} $.

    \begin{enumerate}
    \item
    Define the discrete codebook $\mathcal{C}=\{\underline{x}=\mathbf{G}\underline{y}:\underline{y}\in\mathbb{F}_q^k\}$ where all operations are over $\mathbb{F}_q$. Thus, $\underline{x}\in\mathbb{F}_q^n$.
    \item
    Generate the  $\mathbb{Z}$-lattice $\Lambda_{\mathcal{C}}$ as $\Lambda_{\mathcal{C}}\triangleq\{\underline{\lambda}\in\mathbb{Z}^n:\underline{\lambda}\mod{q}\in \mathcal{C}\}$, where the $\mod$ operation is applied to each component of $\underline{\lambda}$.
    \item
    Scale $\Lambda_{\mathcal{C}}$ with ${q}^{-1}$ to obtain $\Lambda={q}^{-1}\Lambda_{\mathcal{C}}$.
    \end{enumerate}
    We would like to note that only the first two steps that we have stated in Construction A  is required to build a lattice, since the third step simply scales the lattice.
    However when Erez \textit{et. al.} prove the existence of lattices built with Construction A that are good for covering in \cite{erez2005lattices}, they keep $r_\Lambda^{\text{eff}}$
    approximately constant as $n\rightarrow\infty$ and $q\rightarrow\infty$, which is possible only if the third step is used for scaling the lattice.

    \subsection{Nested $\mathbb{Z}$-lattices obtained from Construction-A \cite{erez2004achieving}}\label{sec: nestedZ}

    Let $\Lambda$ be an $n$-dimensional $\mathbb{Z}$-lattice obtained through Construction-A with a corresponding generator matrix $\mathbf{B}$. For
    a given $\mathbf{G}\in\mathbb{F}_q^{n\times k}$, denote $\Lambda^{\prime}$ as the corresponding $\mathbb{Z}$-lattice obtained through Construction-A
    using $\mathbf{G}$ as the generator matrix of the underlying linear code. Generate the $\mathbb{Z}$-lattice $\Lambda_f$ as $\Lambda_f=\mathbf{B}\Lambda^{\prime}$.
    It can be observed that $\Lambda\subset\Lambda_f$ with a coding rate of $\frac{k}{n}\log{q}$.



\section{Compute-and-Forward with $\mathbb{Z}$-lattices}\label{cfgausinf}
One way to implement network coding for the system model considered in this paper is for each relay to decode to $\underline{w}_l$ individually, then form $\underline{f}_m$ and forward it through the network, which is commonly referred as
decode-and-forward. As the number of source nodes $L$ increase, decode-and-forward is limited by self-interference since other transmitted messages are treated as noise
when decoding to $\underline{w}_l$ individually. Therefore, one way to mitigate the effect of self-interference would be for relay $m$ to directly decode to $\underline{f}_m$ from $\underline{y}_m$
instead of decoding to $\underline{w}_l$'s individually. Such an approach is commonly referred to as compute-and-forward, which was introduced by Nazer and Gastpar in \cite{nazer2011CF} and results in achieving substantially higher rates than other forwarding paradigms such as amplify-and-forward, decode-and-forward, compress-and-forward in many situations.


In \cite{nazer2011CF}, Nazer and Gastpar use nested lattice codes to implement the compute-and-forward paradigm. Since lattices are closed under integer combinations, the relays attempt to decode to a linear combination of codewords with integer coefficients. This can then be shown to correspond to decoding linear combinations over the finite field. We briefly discuss how lattice codes are constructed to implement the compute-and-forward paradigm in \cite{nazer2011CF}.

A fine $\mathbb{Z}$-lattice $\Lambda_f$ and a coarse $\mathbb{Z}$-lattice $\Lambda$  nested in $\Lambda_f$, is constructed as mentioned in Section~\ref{sec: nestedZ} with a coding rate $R=\frac{k}{n}\log{q}$. If $\Lambda$ is simultaneously good for covering and good for AWGN channel coding, it follows that $\Lambda_f$ is good for AWGN channel coding \cite{erez2004achieving}. Both $\Lambda$ and $\Lambda_f$ are scaled such that $\sigma^2_\Lambda=P/2$. Following this, the lattice codebook $\Lambda_f\cap\mathcal{V}_\Lambda$ is constructed.

Source node $l$ partitions its information vector $\underline{w}_l\in\mathbb{F}_q^{2k}$ into $\underline{w}_l^R,\underline{w}_l^I\in\mathbb{F}_q^k$, and maps them to lattice codewords $\underline{t}_l^R,\underline{t}_l^I\in\Lambda_f\cap\mathcal{V}$, respectively, via a bijective mapping $\tilde{\psi}$,
\begin{align}
    \tilde{\psi}(\underline{w})=\left[\mathbf{B}q^{-1}g(\mathbf{G}\underline{w})\right],
\end{align}
where $\underline{w}\in\mathbb{F}_q^k$, and $g$ is the trivial bijective mapping between $\{0,1,\cdots,q-1\}$ and $\mathbb{F}_q$. Hence, $\underline{t}_l^R=\tilde{\psi}\left(\underline{w}_l^R\right),\underline{t}_l^I=\tilde{\psi}\left(\underline{w}_l^I\right)$. It then constructs dither vectors $\underline{d}_l^R,\underline{d}_l^I$, which are uniformly distributed within $\mathcal{V}$ and subtracts these dither vectors from the lattice codewords $\underline{t}_l^R,\underline{t}_l^I$, respectively, and transmits the following:
\begin{align}
    \underline{x}_l=\left(\left[\underline{t}_l^R-\underline{d}_l^R\right]\mod{\Lambda}\right)+j\left(\left[\underline{t}_l^I-\underline{d}_l^I\right]\mod{\Lambda}\right).
\end{align}
Recall that given a channel coefficient vector $\underline{h}_m\in\mathbb{C}^L$, relay $m$ observes
\begin{align}
    \underline{y}_m=\sum_{l=1}^L h_{ml}\underline{x}_l+\underline{z}_m.
\end{align}
The relay approximates $\underline{h}_m$, in some sense, by a Gaussian integer vector $\underline{a}_m\in\mathbb{Z}[i]^L$ and its goal will be to recover the following:
\begin{align}
\underline{v}_m^R=\left[\sum_{l=1}^L \left[\Re\left(a_{ml}\right)\underline{t}_l^R-\Im\left( a_{ml}\right)\underline{t}_l^I\right] \right]\mod{\Lambda},\\
\underline{v}_m^I=\left[\sum_{l=1}^L \left[\Im\left(a_{ml}\right)\underline{t}_l^R+\Re\left( a_{ml}\right)\underline{t}_l^I\right] \right]\mod{\Lambda}.
\end{align}
It proceeds by removing the dithers and scaling the observation with $\alpha_m$ and therefore,
\begin{align}\label{eq: gaus_mmse1}
\underline{\tilde{y}}_m^R&=\Re\left(\alpha_m\underline{y}_m\right)+\sum_{l=1}^{L}\Re\left(a_{ml}\right)\underline{d}_l^R-\Im\left(a_{ml}\right)\underline{d}_l^I\nonumber\\
                         &=\underline{v}_m^R+\underline{z}_{eq,m}^R,
\end{align}
and
\begin{align}\label{eq: gaus_mmse2}
\underline{\tilde{y}}_m^I&=\Im\left(\alpha_m\underline{y}_m\right)+\sum_{l=1}^{L}\Im\left(a_{ml}\right)\underline{d}_l^R+\Re\left(a_{ml}\right)\underline{d}_l^I\nonumber\\
&=\underline{v}_m^I+\underline{z}_{eq,m}^I,
\end{align}
where $\alpha_m$ is the MMSE scaling coefficient that minimizes the variance of $\underline{z}_{eq,m}^R+j\underline{z}_{eq,m}^I$. The relay quantizes $\underline{\tilde{y}}_m^I,\underline{\tilde{y}}_m^R$ to the closest lattice points in the fine lattice $\Lambda_f$ modulo the coarse lattice $\Lambda$ and estimates the following:
\begin{align}
\underline{\hat{v}}_m^R=\left[Q\left( \tilde{\underline{y}}_m^R \right) \right]\mod{\Lambda},\\
\underline{\hat{v}}_m^I=\left[Q\left( \tilde{\underline{y}}_m^I \right) \right]\mod{\Lambda},
\end{align}
where $Q$ denotes the quantization with respect to $\Lambda_f$. Finally, the relay maps $\underline{\hat{v}}_m^R$ and $\underline{\hat{v}}_m^I$ to $\underline{\hat{f}}_m^R$ and $\underline{\hat{f}}_m^I$, respectively, via $\tilde\psi^{-1}$,
\begin{align}
\tilde\psi^{-1}(\underline{v})=\left(\mathbf{G}^{T}\mathbf{G}\right)^{-1}\mathbf{G}^{T}g^{-1}\left(q\left(\left[\mathbf{B}^{-1}\underline{v}\mod{\Lambda}\right]\right)\right),
\end{align}
where $\underline{v}\in\mathbb{F}_q^n$. Hence,

\begin{align}
\tilde\psi^{-1}\left(\underline{\hat{v}}_m^R\right)&=\underline{\hat{f}}_m^R=\bigoplus_{l=1}^L\left( b_{ml}^R\underline{\hat{w}}_l^R\oplus\left(-b_{ml}^I \right)\underline{\hat{w}}_l^I \right), \\
\tilde\psi^{-1}\left(\underline{\hat{v}}_m^I\right)&=\underline{\hat{f}}_m^I=\bigoplus_{l=1}^L\left( b_{ml}^I\underline{\hat{w}}_l^R\oplus\left(b_{ml}^R \right)\underline{\hat{w}}_l^I \right),
\end{align}
where
\begin{align}
b_{ml}^R&=\Re\left(a_{ml}\right)\mod{q},\\
b_{ml}^I&=\Im\left(a_{ml}\right)\mod{q}.
\end{align}
Note that both $[\underline{b}_1^R,\ldots,\underline{b}_{M}^R]$ and $[\underline{b}_1^I,\ldots,\underline{b}_{M}^I]$ are required to be full rank so that decoding each $\underline{w}_l^R,\underline{w}_l^I$ at the final destination is feasible.

In \cite{nazer2011CF}, Nazer and Gastpar show the following theorem using the coding scheme we have described in this section.

\begin{theorem}[Nazer and Gastpar]\label{thm:Nazer}
At relay $m$, given $\underline{h}_m\in\mathbb{C}^L$ and ${\underline{a}}_m\in\mathbb{Z}[i]^L$, a computation rate of
\begin{align}\label{cfrate}
    \mathcal{R}(\underline{h}_m,\underline{{a}}_m)=\log^{+}\left(\left(\|\underline{{a}}_m\|^{2}-
    \frac{P|\underline{h}_m^{H}\underline{{a}}_m|^2}{1+P\|\underline{h}_m\|^2}\right)^{-1}\right),
\end{align}
is achievable.
\end{theorem}

Given $\mathbf{H}$ and assuming that the relays do not cooperate with each other, each relay would attempt to pick an integer vector $\underline{a}_m$
that maximizes its individual computation rate, i.e. $\underline{a}_m=\underset{\underline{a}\in\mathbb{Z}[i]^L}{\arg\max}~\mathcal{R}(\underline{h}_m,\underline{{a}}_m)$
in order to maximize $\mathcal{R}\left(\mathbf{H},\mathbf{A} \right)$.

\section{Compute-and-Forward with Lattices over Eisenstein Integers}\label{CFeisinf}
The main result in this section is that for some channel realizations, higher information rates than those in Theorem~\ref{thm:Nazer} are achievable. The improved information rate is obtained by considering nested lattices over Eisenstein integers which allow the $m$th relay to decode a linear combination of the form $\sum_{l=1}^{L}a_{ml} \underline{t}_l$, where $a_{ml} \in \mathbb{Z}[\omega]$. This result is made precise in Theorem~\ref{thm:Eisen}.

%

    One of the key challenges in proving this achievability result is to show the existence of nested lattices over Eisenstein integers, which we refer to as $\mathbb{Z}[\omega]$-lattices, where the coarse lattice is good for covering and the fine lattice can achieve the Poltyrev limit. We would like to note that, we do not prove the existence of $\mathbb{Z}[\omega]$-lattices that are good for AWGN channel coding, i.e. lattices for which the error probability can be bounded exponentially in $n$, in this paper. Furthermore, we do not require the coarse lattice in the sequence of nested lattices to be simultaneously good
 for AWGN channel coding and good for covering. In order to state our main theorem, it suffices to show the existence of nested $\mathbb{Z}[\omega]$-lattices where the coarse lattice
is good for covering and the fine lattice can achieve the Poltyrev limit. A similar result is obtained in \cite{Ordentlich}, where the coarse lattice is chosen to be good only for quantization and the fine
lattice to be good for AWGN channel coding in order to achieve $\frac{1}{2}\log(1+SNR)$ using lattice codes for the point-to-point AWGN channel.

    In what follows, we first provide some preliminaries about Eisenstein integers and summarize Construction A for $\mathbb{Z}[\omega]$-lattices. Afterwards, we show that nested $\mathbb{Z}[\omega]$-lattices where the coarse lattice is good for quantization and the fine lattice achieves the Poltyrev limit can be obtained through Construction A. The existence result can then be used to prove Theorem~\ref{thm:Eisen}, which is the main result of this paper. Since $\mathbb{Z}[\omega]$ quantizes $\mathbb{C}$ better than $\mathbb{Z}[i]$, on the average (over the channel realizations), higher information rates are achievable by using $\mathbb{Z}[\omega]$-lattices compared to using $\mathbb{Z}$-lattices. The superiority of the proposed scheme will be further confirmed in Section~\ref{Sec: num_res} where we provide numerical results to compare the outage performance and error-correcting performance of lattices over natural integers and lattices over Eisenstein integers in compute-and-forward.
\subsection{Preliminaries: Eisenstein Integers}

An Eisenstein  integer is a complex number of the form $a+b\omega$ where $a,b\in\mathbb{Z}$ and $\omega=-\frac{1}{2}+j\frac{\sqrt{3}}{2}$. The ring of Eisenstein integers $\mathbb{Z}[\omega]$ is a principal ideal domain, i.e, a commutative ring without zero divisors where every ideal can be generated by a single element. Other well-known principal ideal domains are $\mathbb{Z}$ and $\mathbb{Z}[i]$. A $unit$ in $\mathbb{Z}[\omega]$ is one of the following:$\{ \pm 1, \pm \omega, \pm \omega^2 \}$. An Eisenstein integer $\varrho$ is an Eisenstein prime if either one of the following mutually exclusive conditions hold \cite{quarternions}:
\begin{enumerate}
\item
$\varrho$ is equal to the product of a unit and any natural prime congruent to $2\mod 3$.
\item
$|\varrho|^2=3$ or $|\varrho|^2$ is  any natural prime congruent to $1\mod 3$.
\end{enumerate}
An $n$-dimensional $\mathbb{Z}[\omega]$-lattice can be written in terms of a complex lattice generator matrix $\mathbf{B}\in\mathbb{C}^{n \times k}$:
\begin{align}
\Lambda=\{ \underline{\lambda}=\mathbf{B}\underline{e}:\underline{e}\in \mathbb{Z}[\omega]^k \}
\end{align}

\subsection{Construction A for $\mathbb{Z}[\omega]$-lattices}\label{sec:ConstrAEis}
{\black Let $\varrho$ be an Eisenstein prime with $|\varrho|^2=q$. Since $\mathbb{Z}[\omega]$ is a principal ideal domain, $\varrho \mathbb{Z}[\omega]$ is an ideal of $\mathbb{Z}[\omega]$ and together they form the quotient ring $\mathbb{Z}[\omega]/\varrho \mathbb{Z}[\omega]$. Moreover, since $\varrho$ is an Eisenstein prime, $\varrho \mathbb{Z}[\omega]$ is a prime ideal and hence a maximal ideal (a property for principal ideal domains). Thus, the quotient ring is isomorphic to a field
\begin{equation}
    \mathbb{Z}[\omega]/\varrho \mathbb{Z}[\omega] \cong \mathbb{F}_{q}.
\end{equation}
i.e., there exists a ring isomorphism $\sigma:\mathbb{Z}[\omega]/\varrho \mathbb{Z}[\omega]\rightarrow\mathbb{F}_q$ \cite[page 118]{Algebra_Book}. Note that $\mathbb{Z}[\omega]$ is the union of $q$ cosets of $\varrho \mathbb{Z}[\omega]$
  \begin{align}
  \mathbb{Z}[\omega]=\underset{s\in\mathcal{S}}{\cup}\left(\varrho\mathbb{Z}[\omega]+s\right)
  \end{align}
where $\mathcal{S}$ represents the set of $q$ coset leaders of $\mathbb{Z}[\omega]/\varrho \mathbb{Z}[\omega]$. One has the canonical ring homomorphism \cite[page 118]{Algebra_Book} $\mod{\varrho\mathbb{Z}[\omega]}:\mathbb{Z}[\omega]\rightarrow\mathbb{Z}[\omega]/\varrho \mathbb{Z}[\omega]$ to homomorphically map an element in $\mathbb{Z}[\omega]$ to its coset leader. Now composing $\mod{\varrho\mathbb{Z}[\omega]}$ and $\sigma$, one obtains the ring homomorphism $\tilde{\sigma}\triangleq\sigma\circ\mod{\varrho\Lambda}:\mathbb{Z}[\omega]\rightarrow\mathbb{F}_q$.}
Note that $\tilde{\sigma}$ can be extended to vectors in a straightforward manner by mapping the elements of the vector componentwise to another vector \cite[page 197]{conway1999sphere}. \color{black}We would like to mention that the aforementioned properties also hold for lattices that are constructed over any other principal ideal domain such as $\mathbb{Z}$ or $\mathbb{Z}[i]$. For example, the $\mod{q}$ operation in Construction A for $\mathbb{Z}$-lattices also provides a ring homomorphism.\color{black} We now define Construction A for $\mathbb{Z}[\omega]$-lattices as follows.

Let $\varrho$ be an Eisenstein prime and $q=|\varrho|^2$. Note that $q$ is either a natural prime or the square of a natural prime. Also let $k,n$ be integers such that $k\leq n$ and let $\mathbf{{G}}\in\mathbb{F}_q^{n\times k} $. Similar to a $\mathbb{Z}$-lattice, a  $\mathbb{Z}[\omega]$-lattice can be obtained by Construction A \cite{conway1999sphere}.


\begin{enumerate}
\item
Define the discrete codebook $\mathcal{C}=\{\underline{x}=\mathbf{G}\underline{y}:\underline{y}\in\mathbb{F}_q^k\}$ where all operations are over $\mathbb{F}_q$. Thus, $\underline{x}\in\mathbb{F}_q^n$.
\item
Generate the $n$-dimensional $\mathbb{Z}[\omega]$-lattice $\Lambda_{\mathcal{C}}$ as $\Lambda_{\mathcal{C}}\triangleq\{\lambda\in\mathbb{Z}[\omega]^n:\tilde\sigma(\lambda)\in \mathcal{C}\}$.
\item
Scale $\Lambda_{\mathcal{C}}$ with ${\varrho}^{-1}$ to obtain $\Lambda={\varrho}^{-1}\Lambda_{\mathcal{C}}$.
\end{enumerate}
Once again, we would like to note that only the first two steps that we have stated in Construction A is required to build a $\mathbb{Z}[\omega]$-lattice. However,due to the fact that we will prove the existence of $\mathbb{Z}[\omega]$-lattices that are good for covering in this paper using similar proof techniques in \cite{erez2005lattices}, we also require the third step which scales the lattice. An example of such a construction with $k=1,n=1, \mathbf{G}=[1]$, $\varrho=2-\sqrt{3}j$, $q=7$ and the corresponding ring homomorphism is shown in Fig.~\ref{fig-constra1}. \color{black} In this figure, the green circles represent $\varrho\mathbb{Z}[\omega]$ and the red lines represent the boundaries of their Voronoi regions. It can be observed that there are exactly $q=|\varrho|^2=7$ lattice points that belong to $\mathbb{Z}[\omega]$ that lie within each Voronoi region of the lattice points that belong to $\varrho\mathbb{Z}[\omega]$. It can also be verified that the mapping (labeling) in Fig.~\ref{fig-constra1} from $\mathbb{Z}[\omega]/\varrho\mathbb{Z}[\omega]$ to $\mathbb{F}_q$ , i.e., $\tilde\sigma$ is indeed a ring homomorphism. \color{black} We would like to note that the lattice in  Fig.~\ref{fig-constra1} is trivially $\mathbb{Z}[\omega]$. Unfortunately, we were not able to provide a less trivial figure with a larger dimensional $\mathbb{Z}[\omega]$-lattice. This is due to the fact that even a two-dimensional $\mathbb{Z}[\omega]$-lattice requires four real dimensions to be drawn, which is not feasible.
\begin{figure}[h]
\centering
\centerline{\includegraphics[width=4.5in]{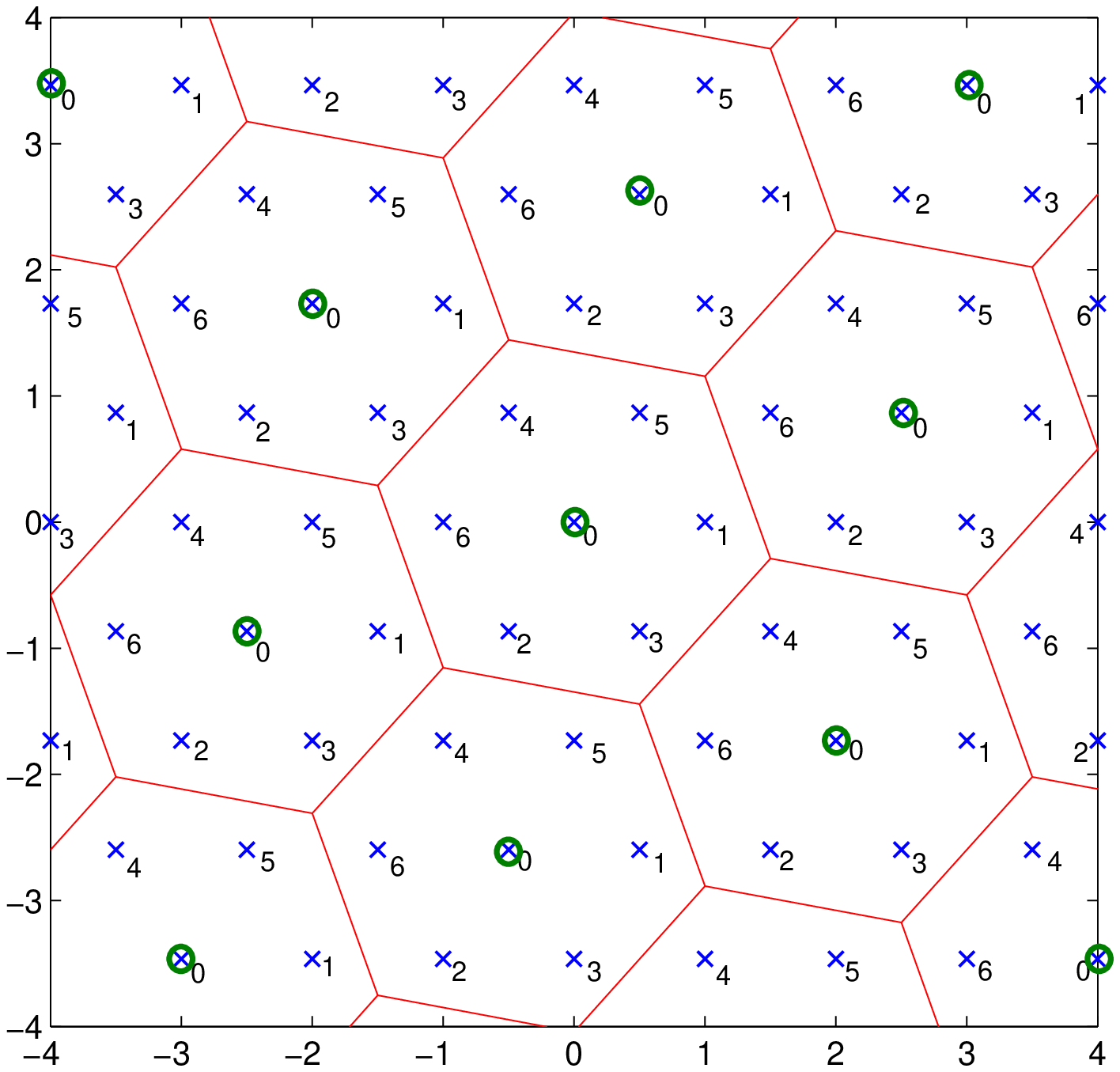}} \caption{ $\Lambda_{\mathcal{C}}$ with $\mathbf{G}=[1]$ and the corresponding ring homomorphism}\label{fig-constra1}
\end{figure}

Given $n,k,q$, we define an $(n,k,q,\mathbb{Z}[\omega])$ ensemble as the set of $\mathbb{Z}[\omega]$-lattices obtained through Construction-A
where for each of these lattices, $\mathbf{G}_{ij}$ are i.i.d with a uniform distribution over $\mathbb{F}_q$.

\begin{theorem}\label{covering}
    A lattice $\Lambda$ drawn from an $(n,k,q,\mathbb{Z}[\omega])$ ensemble, where $k< n$ but grows faster than $\log^2 n$, $q$ is a
natural prime congruent to $1\mod{3}$, and where $k, q$ satisfy
    \begin{align}\label{qkncond11}
        q^k&=\frac{\left(\frac{\sqrt{3}}{2}\right)^n}{V_{\mathcal{B}}\left({r_{\Lambda}^{\text{eff}}} \right)}=\frac{\left(\frac{\sqrt{3}}{2}\right)^n\Gamma\left(n+1\right)}{\pi^n\left(r_{\Lambda}^{\text{eff}} \right)^{2n}}\nonumber\\
        &\approx\sqrt{2n\pi}\left(\frac{\sqrt{3}}{2}\right)^n\left(\frac{2n}{2\exp(1)\left(r_{\Lambda}^{\text{eff}} \right)^{2}}\right)^n,
    \end{align}
    and
    \begin{align}\label{qkncond12}
        r_{min}<r_{\Lambda}^{\text{eff}}< 2 r_{min},
    \end{align}
    where $0<r_{min}<\frac{1}{4}$,
    is good for covering, i.e,
    \begin{align}
        \frac{r_{\Lambda^{{\text{cov}}}}}{r_{\Lambda}^{{\text{eff}}}}\rightarrow 1,
    \end{align}
    in probability as $n\rightarrow\infty$.
\end{theorem}
\begin{proof}
\color{black}We would like to note that the steps we follow in this proof are similar to the proof of Theorem 2 in \cite{erez2005lattices}. The most important differences are as follows. Instead of considering the lattice points that lie within the fundamental Voronoi region of the lattice $\mathbb{Z}^n$, which is an $n$-dimensional unit cube, we consider the lattice points that lie within the fundamental Voronoi region of the lattice $\mathbb{Z}[\omega]^n$, which is an $n$-dimensional hexagon. Furthermore, since we are constrained to $q$ congruent to $1\mod{3}$, Bertrand's postulate is not sufficient to show the existence of such $q$ that satisfies \eqref{qkncond11} and \eqref{qkncond12} as $k$ grows. Therefore, we use the result in \cite{german} to show such prime numbers exist. For the rest of the proof, see Appendix~\ref{sec:GfC}.\color{black}
\end{proof}
We would like to note that a variant of Theorem~\ref{covering} can also be proven for $q$ congruent to $2\mod{3}$, which in this case we can construct $\Lambda$ from linear codes over $\mathbb{F}_{q^2}$.

{ \begin{corollary}\label{quantization}
     A lattice $\Lambda$ drawn from an $(n,k,q,\mathbb{Z}[\omega])$ ensemble, where $k<n$ but grows faster than $\log^2 n$ and where $k,q$ satisfy \eqref{qkncond11} and \eqref{qkncond12} is good for quantization, i.e.,
    \begin{align}
        G\left(\Lambda\right)\rightarrow \frac{1}{2\pi e},
    \end{align}
    in probability as $n\rightarrow\infty$.
\end{corollary}}
\begin{proof}
    It was shown in \cite{Zamir-Feder1996Quan} that a lattice ensemble which is good for covering is necessarily good for quantization. Thus from Theorem \ref{covering}, the result follows.
\end{proof}

\subsection{Nested $\mathbb{Z}[\omega]$-lattices obtained from Construction-A}\label{Sec:NestedZw}
Nested $\mathbb{Z}[\omega]$-lattices can be obtained from Construction-A very similar to $\mathbb{Z}$-lattices as mentioned in Section~\ref{sec: nestedZ}.
The coarse lattice $\Lambda$ is obtained through Construction-A as mentioned in Section~\ref{sec:ConstrAEis} with a corresponding generator matrix $\mathbf{B}$.
For a given $\mathbf{G}\in\mathbb{F}_q^{n\times k}$, denote $\Lambda^{\prime}$ as the corresponding $\mathbb{Z}[\omega]$-lattice obtained through Construction-A
using $\mathbf{G}$ as the generator matrix of the underlying linear code. Generate the $\mathbb{Z}[\omega]$-lattice $\Lambda_f$ as
$\Lambda_f=\mathbf{B}\Lambda^{\prime}$. It can be observed that $\Lambda\subset\Lambda_f$ with a coding rate of $\frac{k}{2n}\log{q}$.
Given $n,~k,~q$ and $\Lambda$ where $\Lambda$ is a $\mathbb{Z}[\omega]$-lattice obtained from Construction-A, we define the
$\left(n,k,q,\Lambda,\mathbb{Z}[\omega]\right)$ ensemble as the set of lattices obtained from $\Lambda$ and Construction-A as previously mentioned where
for each of these lattices, the elements of the generator matrix of the underlying linear code $\mathbf{G}_{ij}$ is i.i.d with a  uniformly distribution
over $\mathbb{F}_q$.

\begin{theorem}\label{GfACC}
There exists a pair of nested $\mathbb{Z}[\omega]$-lattices where the coarse lattice is good for covering
and the fine lattice achieves the Poltyrev limit.
\end{theorem}

\begin{proof}
For this proof, we build nested $\mathbb{Z}[\omega]$-lattices as mentioned above. Using our result
from Theorem~\ref{covering}, we pick a coarse lattice $\Lambda$ which is good for covering. We then pick $\Lambda_f$ from
the $\left(n,k,q,\Lambda,\mathbb{Z}[\omega]\right)$ ensemble
as described in Section~\ref{Sec:NestedZw} and show that the Minkowski-Hlawka theorem can be proven for this ensemble \cite{loeliger}. \color{black} We would like to note that the steps we follow are very similar to the steps followed in \cite{loeliger}. Some of the important differences are as follows. Since we are constructing $\mathbb{Z}[\omega]$-lattices, we consider the fundamental Voronoi region of the lattice $\mathbb{Z}[\omega]^n$ which has a volume of $\left(\frac{\sqrt{3}}{2}\right)^n$. Therefore this should be taken into account when $\text{Vol}\left(\mathcal{V}_{\Lambda_f}\right)$ is kept constant as $n\rightarrow\infty$. \color{black} In the detailed proof provided in Appendix~\ref{sec:GFACC},
it can be observed that a lattice $\Lambda_f$ picked from the $\left(n,k,q,\Lambda,\mathbb{Z}[\omega]\right)$ ensemble achieves the Poltyrev limit as long as the generator matrix $\mathbf{B}$ of $\Lambda$ is full rank. We would like to note that this result is a generalized version of what was stated in \cite{loeliger} where $\mathbf{B}$ was assumed to be an identity matrix. One of the consequences of picking an arbitrary full rank matrix $\mathbf{B}$ would be that $\mathcal{V}_\Lambda$ might stretch out in some dimensions while shrinking in others. \color{black} Nonetheless, since the growth of $q$ in Theorem \ref{covering} ensures that $q\rightarrow\infty$, there is exactly one element in the kernel of $\tilde\sigma$ contained in the bounded region, i.e., the left term of \eqref{MH3} vanishes, and the result holds. \color{black}

\end{proof}

Now, we are ready to state the main theorem in the paper.
\begin{theorem}\label{thm:Eisen}
At relay $m$, given $\underline{h}_m$ and $\underline{a}_m$, a computation rate of
\begin{align}\label{cfrateeis}
    \mathcal{R}(\underline{h}_m,\underline{a}_m)=\log^{+}\left(\left(\|\underline{a}_m\|^{2}-
    \frac{P|\underline{h}_m^{H}\underline{a}_m|^2}{1+P\|\underline{h}_m\|^2}\right)^{-1}\right),
\end{align}
where $\underline{a}_{ml}\in\mathbb{Z}[\omega]$, is achievable.
\end{theorem}

\begin{proof}

\color{black}{We would like to note that the steps we follow in this proof are very similar to the proof of Theorem 5 in \cite{nazer2011CF}. Nonetheless, there are some important differences we would like to point out. Since $a_{ml}$ are Eisenstein integers in our framework, their real and imaginary components are not independent and we cannot use a real and imaginary decomposition as in \cite{nazer2011CF}. Therefore, the channel coefficients and channel noise cannot be decomposed into real and imaginary components either. Due to this, we are constrained to employ $\mathbb{Z}[\omega]$-lattices in our framework. Furthermore, in order to obtain  $b_{ml}$ from $a_{ml}$, we use a ring homomorphism $\sigma$, which can be thought of as the equivalent of a modulo operation for $a_{ml}\in\mathbb{Z}$. We would also like to mention that this proof can be trivially extended to the case where information vectors at transmitters have different lengths by considering a sequence of nested lattice codes. We proceed as follows. }

\color{black}

Using the result from Theorem~\ref{GfACC}, a fine $\mathbb{Z}[\omega]$-lattice $\Lambda_f$ and a coarse $\mathbb{Z}[\omega]$-lattice $\Lambda$, which is nested in $\Lambda_f$ with a corresponding coding rate $\frac{R}{2}=\frac{k}{2n}\log{q}$, is chosen such that $\Lambda_f$ achieves the Poltyrev limit and $\Lambda$ is good for covering.  Both $\Lambda$ and $\Lambda_f$ are scaled such that $\sigma^2_\Lambda=P$. Following this, the lattice codebook $\Lambda_f\cap\mathcal{V}_\Lambda$ is constructed.

Source node $l$ maps its information vector $\underline{w}_l\in\mathbb{F}_q^k$, where $q=|\varrho|^2$ and $\varrho$ is an Eisenstein prime, to a lattice codeword $\underline{t}_l\in\Lambda_f\cap\mathcal{V}_\Lambda$, respectively, via a bijective mapping $\psi$,
\begin{align}
\underline{t}_l={\psi}(\underline{w})=\left[\mathbf{B}\varrho^{-1}\sigma^{-1}(\mathbf{G}\underline{w})\right],
\end{align}
where $\sigma$ was defined in Section~\ref{sec:ConstrAEis}. It then constructs a dither vector $\underline{d}_l$, which is uniformly distributed
within $\mathcal{V}_\Lambda$ and subtracts this dither vector from the lattice codeword $\underline{t}_l$ and transmits the following:
\begin{align}
\underline{x}_l=\left[\underline{t}_l-\underline{d}_l\right]\mod{\Lambda}.
\end{align}
Given a channel coefficient vector $\underline{h}_m\in\mathbb{C}^L$, relay $m$ observes
\begin{align}
\underline{y}_m=\sum_{l=1}^L h_{ml}\underline{x}_l+\underline{z}_m.
\end{align}
The relay approximates $\underline{h}_m$, in some sense, by an Eisenstein integer vector $\underline{a}_m\in\mathbb{Z}[\omega]^L$ and its goal will be to recover the following:
\begin{align}
\underline{v}_m=\left[\sum_{l=1}^L \left({a}_{ml}\underline{t}_l\right)\right]\mod{\Lambda}.
\end{align}
It proceeds by removing the dithers and scaling the observation with $\alpha_m$,  and therefore,
\begin{align}\label{eis_noise}
\underline{\tilde{y}}_m&=\alpha_m\underline{y}_m+\sum_{l=1}^{L}{a}_{ml}\underline{d}_l,
\end{align}
where $\alpha_m$ is the MMSE coefficient.

Then $\underline{\tilde{y}}_m$ is quantized to the closest lattice point in the fine lattice $\Lambda_f$ modulo the coarse lattice $\Lambda$ and estimates the following:

\begin{align}
\underline{\hat{v}}_m=\left[Q_{\Lambda_f}\left( \tilde{\underline{y}}_m \right) \right]\mod{\Lambda},
\end{align}
where $Q_{\Lambda_f}$ denotes the quantization with respect to $\Lambda_f$. \color{black}The remaining steps of the proof would be identical to the steps in the proof of Theorem 5 in \cite{nazer2011CF} with the only difference being as follows. \color{black} The relay maps $\underline{\hat{v}}_m$ to $\underline{\hat{f}}_m$ via $\psi^{-1}$, where
\begin{align}
\psi^{-1}\left(\underline{\hat{v}}_m\right)=\underline{\hat{f}}_m=\left(\mathbf{G}^{T}\mathbf{G}\right)^{-1}\mathbf{G}^{T}\sigma\left(\varrho\left(\left[\mathbf{B}^{-1}\underline{\hat{v}}_m\mod{\Lambda}\right]\right)\right)=\bigoplus_{l=1}^L b_{ml}\underline{\hat{w}}_l,
\end{align}
and $b_{ml}=\sigma\left(a_{ml}\right)$.

Due to the fact that $\Lambda$ is good for covering and the dithers are uniformly distributed in $\mathcal{V}_{\Lambda}$, the probability density function of the equivalent noise $\underline{z}_{eq,m}$ is upper-bounded by a zero-mean complex Gaussian with a variance that approaches $|\alpha_m|^2+P||\alpha_m \underline{h}_m-\underline{a}_m||^2$ multiplied by a constant as $n\rightarrow\infty$ (\cite[Lemma 8]{nazer2011CF}). We would like to note that the error probability $\text{Pr}\left(\underline{z}_{eq}\not\in\mathcal{V}_{\Lambda_f}\right)$ goes to zero as $n\rightarrow\infty$, however this decay
is not necessarily exponential in $n$, since we have only proven the existence of $\mathbb{Z}[\omega]$-lattices which achieve the Poltyrev limit
and this result does not provide information about the error exponents of such lattices. Nonetheless, it is sufficient to achieve the
computation rate in \eqref{cfrateeis}.
\end{proof}
%

\color{black}

Given $\mathbf{H}$ and assuming that the relays do not cooperate with each other, each relay would attempt to pick $\underline{a}_m\in\mathbb{Z}[\omega]^L$
that maximizes its individual computation rate, i.e. $\underline{a}_m=\underset{\underline{a}\in\mathbb{Z}[\omega]^L}{\arg\max}~\mathcal{R}(\underline{h}_m,\underline{{a}}_m)$
in order to maximize $\mathcal{R}\left(\mathbf{H},\mathbf{A} \right)$. A straightforward method to determine the optimal $\underline{a}_m$ would be to employ an exhaustive search over all $\underline{a}_m$ that satisfies $\|\underline{a}_m\|^2<1+\|\underline{h}_m\|^2P$ (\cite[Lemma 1]{nazer2011CF}). One major challenge in the compute-and-forward paradigm is that for large $P$ and $L$, exhaustively searching optimal $\underline{a}_m$ becomes infeasible. Nonetheless, this problem can be molded into a different form which enables the utilization of much more efficient algorithms (see \cite{feng} for $\mathbb{Z}[i]$ and \cite{other_eis} for $\mathbb{Z}[\omega]$ for example.) In the following subsection, we review this approach for the sake of completeness.

\subsection{An efficient algorithm for choosing $\underline{a}_m$}\label{Sec: Choose a}

As can be seen in (\cite{nazer2011CF}), upon scaling $\underline{y}_m$ with the MMSE coefficient $\alpha_m$, the effective noise variance at relay $m$, which we denote as $\sigma_{\text{eff},m}^2$,
can be computed as

\begin{align}\label{eq: noise eff}
\sigma_{\text{eff},m}^2=|\alpha_m|^2+P\|\alpha_m\underline{h}_m-\underline{a}_m\|^2,
\end{align}
where
\begin{align}\label{eq: mmse coeff}
\alpha_m=\frac{P\underline{h}_m^{H}\underline{a}_m}{1+\|\underline{h}_m\|^2}.
\end{align}
Furthermore, the achievable computation rate at each relay can be expressed in terms of $P$ and $\sigma_{\text{eff},m}^2$ as
\begin{align}
\mathcal{R}\left(\underline{h}_m,\underline{a}_m\right)=\log^{+}\left(\frac{P}{\sigma_{\text{eff},m}^2}\right).
\end{align}
Therefore,
\begin{align}
\underset{\underline{a}_m\in\mathbb{Z}[\omega]^L}{\arg\max}~\mathcal{R}\left(\underline{h}_m,\underline{a}_m\right)=\underset{\underline{a}_m\in\mathbb{Z}[\omega]^L}{\arg\min}~\sigma_{\text{eff},m}^2.
\end{align}

We now take a closer look at $\sigma_{\text{eff},m}^2$. Substituting \eqref{eq: mmse coeff} in \eqref{eq: noise eff}, it can be observed that
\begin{align}
\sigma_{\text{eff},m}^2&=P\underline{a}_m^{H} \underline{a}_m-\frac{P^2 \underline{a}_m^{H} \underline{h}_m \underline{h}_m^{H} \underline{a}_m }{1+P\|\underline{h}_m\|^2}\nonumber\\
&=P\underline{a}_m^{H}\left(\mathbf{I}-\frac{P\underline{h}_m\underline{h}_m^{H}}{1+P\|\underline{h}_m\|^2}\right)\underline{a}_m
\end{align}
Due to the Matrix Inversion Lemma \cite{MIL},
\begin{align}
\mathbf{I}-\frac{P\underline{h}_m\underline{h}_m^{H}}{1+P\|\underline{h}_m\|^2}=\left(I+P\underline{h}_m\underline{h}_m^{H}\right)^{-1},
\end{align}
and $\sigma_{\text{eff},m}^2$ can be expressed as
\begin{align}
\sigma_{\text{eff},m}^2=P\underline{a}_m^{H}\left(I+P\underline{h}_m\underline{h}_m^{H}\right)^{-1}\underline{a}_m.
\end{align}
Note that $\left(I+P\underline{h}_m\underline{h}_m^{H}\right)$, which we denote as $\mathbf{S}$, is a Hermitian matrix. Therefore, the singular value decomposition of $\mathbf{S}$ can be expressed as $\mathbf{V}\mathbf{D}\mathbf{V}^{H}$, where $\mathbf{D}$ is a diagonal matrix which has the eigenvalues of $\mathbf{S}$ as non-zero entries and $\mathbf{V}$ is an orthogonal matrix which has the corresponding eigenvectors of $\mathbf{S}$ in its columns. Hence,
\begin{align}
\sigma_{\text{eff},m}^2&=P\underline{a}_m^{H}\left(\mathbf{V}\mathbf{D}^{-1}\mathbf{V}^{H}\right)\nonumber\\
&=P\|\mathbf{D}^{-1/2}\mathbf{V}^{H}\underline{a}_m\|^2,
\end{align}
and therefore it can be concluded that
\begin{align}\label{eq: SVP}
\underset{\underline{a}_m\in\mathbb{Z}[\omega]^L}{\arg\min}~\sigma_{\text{eff},m}^2=\underset{\underline{a}_m\in\mathbb{Z}[\omega]^L}{\arg\min}~\|\mathbf{D}^{-1/2}\mathbf{V}^{H}\underline{a}_m\|^2.
\end{align}
Thus, the search in \eqref{eq: SVP} is equivalent to finding the non-zero minimal Euclidean norm point
generated by $\mathbf{D}^{-1/2}\mathbf{V}^{H}$ as a $\mathbb{Z}[\omega]$-lattice,
which is commonly referred to as the shortest vector problem (SVP). For reasonable values of $L$,
e.g. $L \le 32$, one of the shortest lattice vectors can be found via a Pohst enumeration or
a Schnorr-Euchner enumeration in a way similar to standard sphere decoding \cite{viterbo99}\cite{Agrell}.
A polynomial-time method to approximate \eqref{eq: SVP} is based on LLL reduction \cite{LLL}.
For our lattices, an LLL over $\mathbb{Z}[\omega]$ should be used as devised by Napias for
Euclidean rings \cite{Napias} including both $\mathbb{Z}[i]$ and $\mathbb{Z}[\omega]$.
Also in \cite{Sakzad}, LLL has been proposed in a different methodology
with no singular value decomposition of $\mathbf{S}$.
Finding approximately optimal $\underline{a}_m$ efficiently is an active research area. The interested reader is referred to \cite{viterbo} and the references therein.

\color{black}

\section{Numerical Results}\label{Sec: num_res}

In this section, we present some numerical results on the achievable computation rates with $\mathbb{Z}[\omega]$-lattices and compare them to the maximum achievable rates with $\mathbb{Z}$-lattices. We consider the case of $L=2$ transmitters and there is $M=1$ relay. For a given channel coefficient vector $\underline{h}$, let $\mathcal{R}_E(\underline{h})$ and $\mathcal{R}_G(\underline{h})$, denote the maximum achievable rate using $\mathbb{Z}[\omega]$-lattices and $\mathbb{Z}$-lattices, respectively, i.e.,
\begin{eqnarray}\label{eq: eis_max_rate}
\mathcal{R}_E(\underline{h},P)=\underset{\underline{a}\in\mathbb{Z}[\omega]^2}{\max}\log^{+}\left(\left(\|\underline{a}\|^{2}-\frac{P|\underline{h}^{H}\underline{a}|^2}{1+P\|\underline{h}\|^2}\right)^{-1}\right),
\end{eqnarray}
and
\begin{eqnarray}\label{eq: gauss_max_rate}
\mathcal{R}_G(\underline{h},P)=\underset{\underline{\tilde{a}}\in\mathbb{Z}[i]^2}{\max}\log^{+}\left(\left(\|\underline{\tilde{a}}\|^{2}-\frac{P|\underline{h}^{H}\underline{\tilde{a}}|^2}{1+P\|\underline{h}\|^2}\right)^{-1}\right).
\end{eqnarray}

\begin{figure}[h]
\centering
\centerline{\includegraphics[width=4.5in]{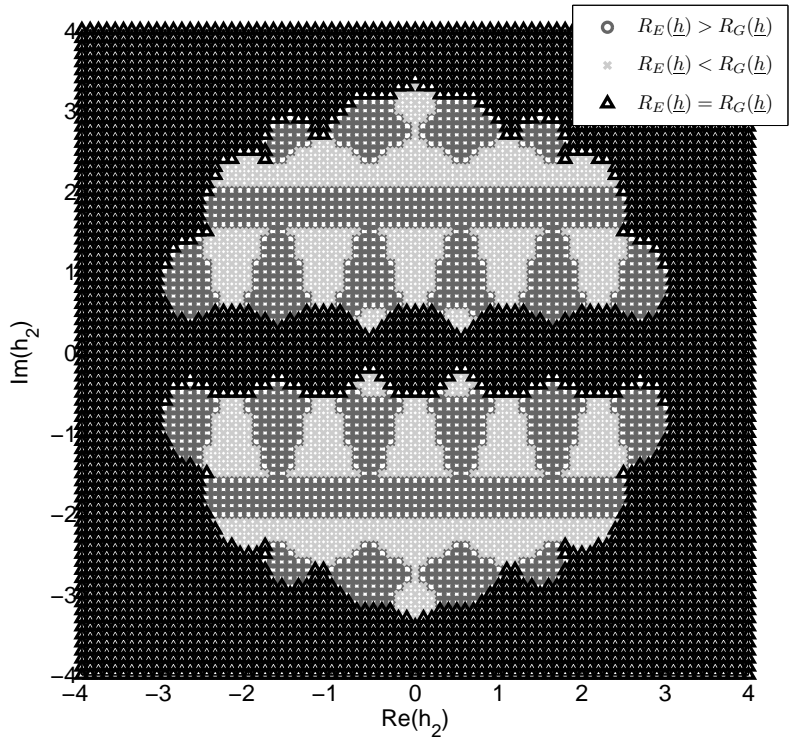}} \caption{Regions of $\Re\left(h_2\right),\Im\left(h_2\right)$ where $\mathcal{R}_G(\underline{h},P)>\mathcal{R}_E(\underline{h},P)$, $\mathcal{R}_G(\underline{h},P)<\mathcal{R}_E(\underline{h},P)$ or $\mathcal{R}_G(\underline{h},P)=\mathcal{R}_E(\underline{h},P)$: SNR=10 dB  }\label{fig-Eis and Gauss Comp}
\end{figure}

In Fig.~\ref{fig-Eis and Gauss Comp}, we fix $h_1=1$ and choose $h_2$ such that $\Re({h_2}),\Im({h_2})\in[-4,4]$. We would also like to note that we do not impose a probability distribution on $h_2$.
For each pair $(h_1=1,h_2)$, we plot the region where $\mathcal{R}_G(\underline{h})>\mathcal{R}_E(\underline{h})$, $\mathcal{R}_G(\underline{h})<\mathcal{R}_E(\underline{h})$ or $\mathcal{R}_G(\underline{h})=\mathcal{R}_E(\underline{h})$. For the total number of realizations considered, $\mathcal{R}_E>\mathcal{R}_G$, $\mathcal{R}_E<\mathcal{R}_G$.
and $\mathcal{R}_E=\mathcal{R}_G$ for  $22.6\%$, $15.9\%$, and $61.5\%$ of the realizations, respectively. One might expect that $\mathbb{Z}[\omega]$-lattices would attain
a greater maximum achievable rate when $h_2$ is closer to an Eisenstein integer, $\mathbb{Z}$-lattices would attain a greater maximum achievable rate when
$h_2$ is closer to a Gaussian integer and both lattices would achieve the same maximum achievable rate when $h_2$ is closer to a natural integer. However as seen from Fig.~\ref{fig-Eis and Gauss Comp}, other factors
also contribute to the maximum achievable rate. For example when $\|h_2\|\gg\|h_1\|$ or $\|h_2\|\ll\|h_1\|$, the relay chooses $a_1=0,\|a_2\|=1$ or $\|a_1\|=1,\|a_2\|=0$, respectively since treating the other transmitted
signal as noise (decode-and-forward) results in maximum achievable rate. Also, the MMSE scaling coefficient $\alpha$ plays a very important role as seen in \eqref{eq: gaus_mmse1}, \eqref{eq: gaus_mmse2} and \eqref{eis_noise}.
Note that \eqref{eq: eis_max_rate} and \eqref{eq: gauss_max_rate} can be written as

\begin{align}\label{eq: eis_max_rate_2}
\mathcal{R}_E(\underline{h},P)=\underset{\underline{a}\in\mathbb{Z}[\omega]^2}{\max}\log^{+}\left(\frac{1+P\|\underline{h}\|^2}{\|\underline{a}\|^{2}+P\left(\|\underline{a}\|^{2}|\underline{h}\|^{2}-|\underline{h}^{H}\underline{a}|^2\right)} \right)\nonumber\\
\end{align}
and
\begin{align}\label{eq: gauss_max_rate_2}
\mathcal{R}_G(\underline{h},P)=\underset{\underline{\tilde{a}}\in\mathbb{Z}[i]^2}{\max}\log^{+}\left(\frac{1+P\|\underline{h}\|^2}{\|\tilde{a}\|^{2}+P\left(\|\underline{\tilde{a}}\|^{2}|\underline{h}\|^{2}-|\underline{h}^{H}\underline{\tilde{a}}|^2\right)} \right),\nonumber\\
\end{align}
respectively.

As one can see from the denominators in \eqref{eq: eis_max_rate_2} and \eqref{eq: gauss_max_rate_2}, it is desirable to align $\underline{a}$ ($\underline{\tilde{a}}$) with $\underline{h}$ as much as possible in order to minimize the second term. However, when
$\underline{h}\not\in\mathbb{Z}[i]^2$,$\underline{h}\not\in{\mathbb{Z}[\omega]}^2$, or the elements of $\underline{h}$ cannot be written as the ratio of Gaussian integers or Eisenstein integers, or $\underline{h}$ is not a rotated version of a Gaussian integer vector
or Eisenstein integer vector, $\|\underline{a}\|\rightarrow\infty$ ($\|\underline{\tilde{a}}\|\rightarrow\infty$)
for perfect alignment. Unfortunately, this results in the first term of the denominator to grow and hence there is a tradeoff. Therefore even though $h_2$ might be closer to an Eisenstein integer (Gaussian integer), i.e. $\underline{h}$ is aligned better with
a vector in $\mathbb{Z}[i]^2$ ($\mathbb{Z}[\omega]^2$), the magnitude of this vector might be too large and thus a larger computation rate may be achieved by choosing $\underline{a}\in\mathbb{Z}[i]^2$ ($\underline{\tilde{a}}\in{\mathbb{Z}[\omega]}^{2}$).

\begin{figure}[h]
\centering
\centerline{\includegraphics[width=4.5in]{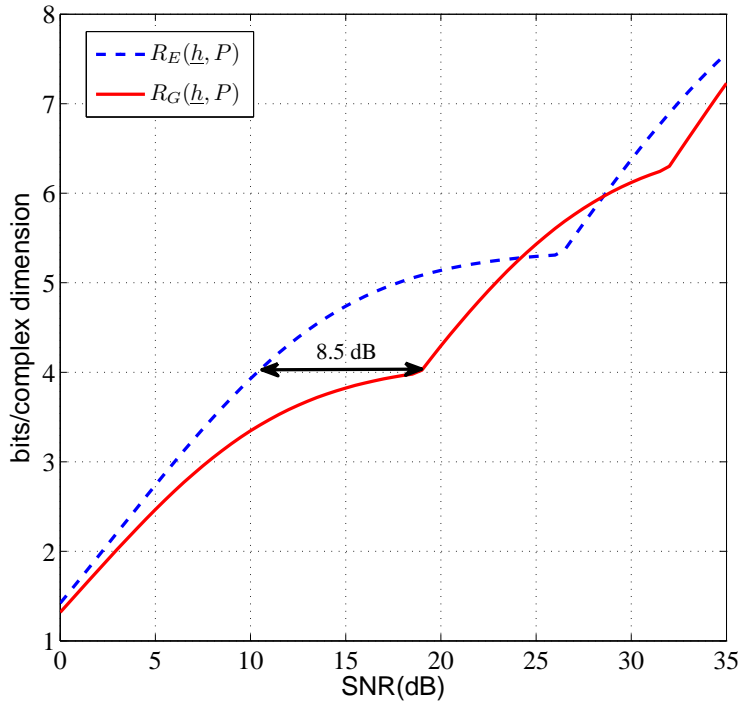}} \caption{A comparison of $\mathcal{R}_E(\underline{h},P)$ and $\mathcal{R}_G(\underline{h},P)$ for $\underline{h}=[1.4193+j0.2916;0.1978+j1.5877]$ }\label{fig-Eis and Gauss Comp 1}
\end{figure}
In Fig.~\ref{fig-Eis and Gauss Comp 1}, we fix the channel realization to be $\underline{h}=[1.4193+j0.2916;0.1978+j1.5877]$ and compare $\mathcal{R}_E(\underline{h},P),~\mathcal{R}_G(\underline{h},P)$ for different SNRs. For this particular $\underline{h}$, it can be observed that $\mathbb{Z}[\omega]$-lattices can achieve substantially higher rates than $\mathbb{Z}$-lattices in the medium SNR regime.  We would like to note that this is not necessarily the case for every channel realization, nonetheless it is a perfect example of how channel realizations affect the performance of $\mathbb{Z}[\omega]$-lattices and $\mathbb{Z}$-lattices. Therefore, a larger number of channel realizations should be considered in order to make a fair comparison of their performance in the average sense.

\subsection{Outage performance comparison of $\mathbb{Z}$-lattices vs. $\mathbb{Z}[\omega]$-lattices in compute-and-forward}
In this subsection, we compare the outage performance lattice codes over $\mathbb{Z}$ and lattice codes over $\mathbb{Z}[\omega]$ for
compute-and-forward. Given a target rate $R_{T}$ and a probability distribution $\mathcal{P}$ on $\underline{h}$, i.e. $\underline{h}\sim\mathcal{P}$,
 we define the outage event of using $\mathbb{Z}$-lattices and $\mathbb{Z}[\omega]$-lattices as $\mathcal{R}_G(\underline{h})<R_{T}$
 and $\mathcal{R}_E(\underline{h})<R_{T}$, respectively. In Fig. \ref{fig-Eis and Gauss Out}, we plot the outage probability with $\mathbb{Z}[\omega]$-lattices
  and $\mathbb{Z}$-lattices as a function of SNR ($P$) where
  $\Re \left(h_1\right),\Im \left(h_1\right),\Re \left(h_2\right),\Im \left(h_2\right)\sim \mathcal{N}(0,1)$.
  We average over 100000 realizations of $\underline{h}$ at each SNR and choose the target rate to be $R_T=1/2\log_{2}{7}$ bits/symbol/Hz.
  As seen in Fig. \ref{fig-Eis and Gauss Out}, there is a 0.4 dB gain from using $\mathbb{Z}[\omega]$-lattices instead of $\mathbb{Z}$-lattices
   in terms of outage performance. We would like to note that this gain comes with no additional computational complexity.

\begin{figure}[h]
\centering
\centerline{\includegraphics[width=4.5in]{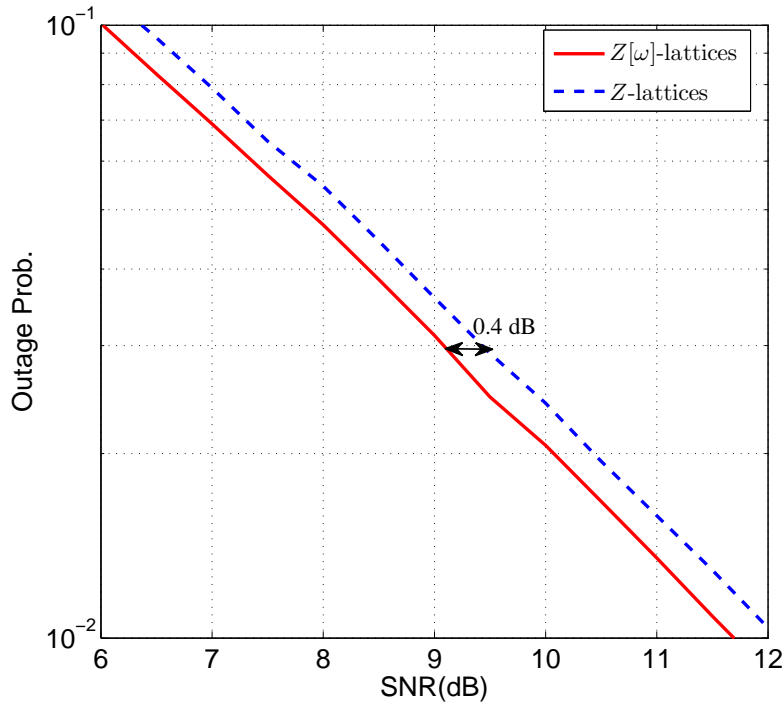}} \caption{Outage Probability of $\mathbb{Z}[\omega]$ Lattices vs $\mathbb{Z}$ Lattices }\label{fig-Eis and Gauss Out}
\end{figure}

\subsection{Error correcting capability of $\mathbb{Z}$-lattices vs. $\mathbb{Z}[\omega]$-lattices in compute-and-forward}

\color{black} In this subsection, we compare the error-correcting capability of lattice codes over $\mathbb{Z}$ and lattice codes over $\mathbb{Z}[\omega]$ for
compute-and-forward. Before we do that, we would like to point out that in general, the nested lattice shaping adopted in the previous sections is very difficult to be implemented. In fact, it is equivalent to the SVP and hence is NP-hard. In practice, one could trade performance for complexity by considering the use of hypercube shaping. Then the proposed scheme would reduce to the concatenation of a linear code over $\mathbb{F}_q$ with a constellation corresponding to a set of minimum energy coset leaders of the quotient ring $\mathbb{Z}[\omega]/\varrho\mathbb{Z}[\omega]$ (or $\mathbb{Z}/q\mathbb{Z}$). In the following, we compare the error-correcting capability for this practical scheme.

In order to construct a lattice code over Eisenstein integers, we have used a rate 1/2, regular (3,6), uniformly distributed edge weight,
length 10000 LDPC code over $\mathbb{F}_{25}$ and mapped each codeword component to the constellation carved from $\mathbb{Z}[\omega]/5 \mathbb{Z}[\omega]$ via a ring homomorphism. In order to construct a lattice code over natural integers, we have used a rate 1/2, regular (3,6), uniformly distributed edge weight, length 10000 LDPC code over $\mathbb{F}_{5}$ and mapped each codeword component to the coset leaders of the quotient ring $\mathbb{Z}/5\mathbb{Z}$, i.e. $\{-2,-1,0,1,2\}$. Note that for the lattice code over natural integers, we consider $\mathbb{F}_{5}$ due to the real and imaginary decomposition. We have generated 100000 channel realizations, used these channel realizations over a range of SNR, and we have plotted the average symbol error probability of these lattice codes for the compute-and-forward framework. As seen in Fig. \ref{fig-Eis and Gauss SER} simulation results show that lattice codes over Eisenstein integers outperform lattice codes over integers by roughly 0.4~dB, which is consistent with our outage simulation results.

\begin{figure}[h]
\centering
\centerline{\includegraphics[width=4.5in]{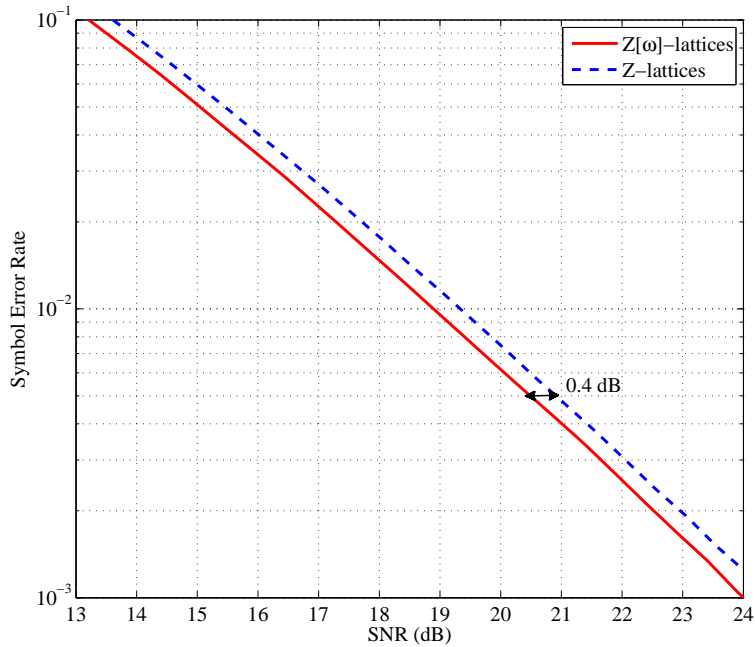}} \caption{Symbol error rate of $\mathbb{Z}[\omega]$ Lattices vs $\mathbb{Z}$ Lattices }\label{fig-Eis and Gauss SER}
\end{figure}

\color{black}

\section{Conclusion}
In this paper, we have shown the existence of lattices over Eisenstein integers that are simultaneously good for quantization and that achieve the Poltyrev limit. \color{black} These lattices were then used to generate lattice codes over Eisenstein integers which were implemented for compute-and-forward and thus enable the relays to decode to linear combinations of lattice points with Eisenstein integer coefficients instead of Gaussian integers. Due to the fact that Eisenstein integers quantize channel coefficients better than Gaussian integers, one can expect an increased achievable computation rate on average. Simulation results suggest that for compute-and-forward, lattice codes over Eisenstein integers provide improved outage performance and error-correcting performance in the average sense compared to lattice codes over integers without the cost of additional computational complexity. \color{black}

\appendix
In this section, we provide the proofs for Theorem~\ref{covering} and Theorem~\ref{GfACC}. We would like to note that the proof techniques used in proving Theorem~\ref{covering} are very similar to those used in \cite{erez2005lattices} and our proof of Theorem~\ref{GfACC} is largely based on the proof in \cite{loeliger}. However, there are a few steps that have to be re-derived since Eisenstein integers are considered. We present the entire proof for the purpose of completeness. We first give some definitions and preliminaries that will be very useful for the proofs.
\subsection{Notations and Definitions for $\mathbb{Z}[\omega]$-lattices}\label{Sec: Appendix Notation}
In \cite[p. 54]{conway1999sphere}, it is stated that an $n$-dimensional complex lattice can be equivalently thought of as a $2n$-dimensional real lattice by the following mapping
\begin{align}\label{ctormap}
[\lambda(1) \cdots\lambda(n)]^{T}\rightarrow[\Re(\lambda(1))\hspace{0.05in}\Im(\lambda(1)) \cdots \Re(\lambda(n))\hspace{0.05in}\Im(\lambda(n))]^{T}\nonumber\\
\end{align}
where the left hand side is an $n$-dimensional complex lattice point and the right hand side is its $2n$-dimensional real representation. Thus we shall consider $n$-dimensional Eisentein lattices as $2n$-dimensional real lattices and use $\mathbb{C}^n$ and $\mathbb{R}^{2n}$ interchangeably.
We shall now introduce the notation that will be used in this section.
\begin{itemize}
\item
$S^{\prime}$: $S\setminus{0}$, where $S$ is any discrete set.
\item
$\mathcal{V}$: Fundamental Voronoi region of the lattice $\mathbb{Z}[\omega]^n$.
\item
GRID: The lattice $\varrho^{-1}\mathbb{Z}[\omega]^n$, where $\varrho$ is an Eisenstein prime.
\item
$\underline{x}^*=\underline{x}\mod\mathcal{V}=\underline{x}\mod \mathbb{Z}[\omega]^n=\underline{x}-Q_{\mathbb{Z}[\omega]^n}\left( \underline{x}\right) $ where
$\underline{x}\in\mathbb{C}^n$.
\item
$\mathcal{A}^{*}=\mathcal{A}\mod \mathcal{V}$, where $\mathcal{A}$ is any set in $\mathbb{C}^n$ and the $\mod \mathcal{V}$ operation is done element-wise.
\item
$\mathcal{A}^{\prime} \triangleq \mathcal{A}\setminus\lbrace0\rbrace$ where $\mathcal{A}\subset\mathbb{R}^n$, $\mathcal{A}\subset\mathbb{C}^n$ or $\mathcal{A}\subset\mathbb{F}_q^n$
\item
$\Lambda$: An $n$-dimensional $\mathbb{Z}[\omega]$-lattice nested in GRID, i.e., $\Lambda \subset\text{GRID}$ .
\item
Vol($\cdot$): Volume of a closed set in $\mathbb{C}^n$, or equivalently volume of a closed set in $\mathbb{R}^{2n}$.
\item
$\text{GRID}^{*}$:  $\text{GRID}\cap\mathcal{V}$.
\item
$\mathcal{B}(r)$:A complex $n$-dimensional, or equivalently real $2n$-dimensional, closed set of points inside a sphere of radius $r$ centered at the origin.
\item
$\Lambda^{*}$: The lattice constellation, i.e. $\Lambda^{*}=\Lambda\cap \mathcal{V}$. Note that $\Lambda^{*}$ can generate $\Lambda$ as follows:
\begin{align}
\Lambda=\Lambda^{*}+\mathbb{Z}[\omega]^n.
\end{align}
\item
$M=|\Lambda^{*}|$: Cardinality of the lattice constellation.
\item
$\Lambda_{i}^{*}$: A point in $\Lambda^{*}$, $i\in\{ 0,\cdots,M-1 \}$.
\end{itemize}
Note that by our construction, the lattices chosen from the $(n,k,q,\mathbb{Z}[\omega])$-lattice ensemble are periodic modulo the region $\mathcal{V}$.
Thus we can restate all the properties of our lattice in terms of the lattice constellation $\Lambda^{*}$ that lies within $\mathcal{V}$.
The $(n,k,q,\mathbb{Z}[\omega])$-lattice ensemble has the following properties:
\begin{enumerate}
\item
$\Lambda_{0}^{*}=\underline{0}$ deterministically.
\begin{proof}
$\underline{0}$ is always a valid lattice point due to the definition of a lattice and $\underline{0}^{*}=\underline{0}$. Thus the result holds.
\end{proof}
\item
$\Lambda_{i}^{*}$ is distributed uniformly over $\text{GRID}^*$ for $i\in\{1,\cdots,M-1 \}$ where $M=q^k$.
\begin{proof}
Each element of $\mathbf{G}$ is chosen uniformly over $\mathbb{F}_q$, therefore each codeword of the underlying linear code
is distributed uniformly over $\mathbb{F}_q^n$. Due to last step in Construction A
in Section~\ref{sec:ConstrAEis} where the lattice is scaled with $\varrho^{-1}$ and the ring homomorphism $\tilde\sigma$, the result holds.
\end{proof}
\item
The difference $({\Lambda_{i}^{*}-\Lambda_{l}^{*}})^*$ is uniformly distributed over $\text{GRID}^{*}$ for all $i\neq j$.
\begin{proof}
This result holds due to the previous property and the definition of the $*$ operation.
\end{proof}
\item
$|\Lambda^*|=q^k$ with high probability if $n-k\rightarrow\infty$
\begin{proof}
\begin{align}
\text{Pr}\{\text{rank}(\mathbf{G})<k\}&\leq\sum_{\underline{c}\neq\underline{0}}\text{Pr}\left\{\sum_{i=1}^k c_i \mathbf{G}_i=\underline{0} \right\}\nonumber\\
&=q^{-n}(q^k-1),
\end{align}
where $c_i$ would be elements of a $k\times1$ coefficient vector $\underline{c}$.
\end{proof}
\end{enumerate}

We shall refer to
$\mathcal{B}(r)^*=\mathcal{B}(r)\mod \mathcal{V}$ as a $\mathcal{V}$-ball. Under the assumption that $r<\frac{1}{2}$, we say that $\left(\Lambda^{*} +\mathcal{B}(r) \right)^*$
is a $\mathcal{V}$-covering if
\begin{align}
\mathcal{V}\subseteq\bigcup_{\underline{\lambda}\in\Lambda^{*}}\left(\underline{\lambda}+\mathcal{B}(r)\right)^{*}.
\end{align}
Note that $\Lambda+\mathcal{B}(r)$ is a covering if and only if $\left(\Lambda^{*} +\mathcal{B}(r) \right)^*$is a $\mathcal{V}$-covering\\

In our lattice ensemble, we will constrain $k<\beta n$ for some $0<\beta<1$. Therefore $\text{Pr}\{\text{rank}(\mathbf{G})\neq k\}$ goes to zero at least exponentially. If $\mathbf{G}$ is full rank, there are $M=q^k$ many codewords that lie in $\mathcal{V}$. Also, an $n$-dimensional $\mathcal{V}$ is known to have a volume of $\left(\frac{\sqrt{3}}{2}\right)^n$. Then the volume of the Voronoi
 region of our lattice is equal to $\left(\frac{\sqrt{3}}{2}\right)^n q^{-k}$. In our analysis very similar to \cite{erez2005lattices}, we will hold the effective
 radius of the Voronoi region of $\Lambda$, denoted as $r_{\Lambda}^{\text{eff}}$ approximately constant as $n\rightarrow\infty$. This implies the following:
\begin{align}\label{qkncond1}
q^k&=\frac{\left(\frac{\sqrt{3}}{2}\right)^n}{V_{\mathcal{B}}\left({r_{\Lambda}^{\text{eff}}} \right)}=\frac{\left(\frac{\sqrt{3}}{2}\right)^n\Gamma\left(n+1\right)}{\pi^n\left(r_{\Lambda}^{\text{eff}} \right)^{2n}}\nonumber\\
&=\sqrt{2n\pi}\left(\frac{\sqrt{3}}{2\left(r_{\Lambda}^{\text{eff}} \right)^{2}}\right)^n\left(\frac{n}{e}\right)^n\left(1+O\left(\frac{1}{n}\right)\right).
\end{align}
Note that $q$ can either be a natural prime congruent to $1\mod 3$ or the square of a natural prime congruent to $2\mod 3$, nonetheless we shall restrict $q$ to be a natural prime congruent to $1\mod 3$ for the sake of simplicity. 
We would like to note that it is not possible to keep $r_{\Lambda}^{\text{eff}}$ constant as $n$ grows since $q$ has to be a natural prime congruent
to $1\mod{3}$ and $k$ has to be an integer. Therefore, we will relax this condition to
\begin{align}\label{qkncond2}
r_{min}<r_{\Lambda}^{\text{eff}}< 2 r_{min},
\end{align}
as $n$ grows, where $0<r_{min}<\frac{1}{4}$. Although we have restricted $q$ to be a natural prime congruent to $1\mod 3$ , with the assumption
 of $k\leq\beta n$ for $\beta<1$, (\ref{qkncond2}) can be satisfied for any large enough $n$ due to the following. Let $q^*$ be the real number
 that satisfies (\ref{qkncond1}) for a radius of $2 r_{min}$. Then, $q^{*^k}=\frac{1}{V_{\mathcal{B}}(\sqrt{\frac{2}{\sqrt{3}}}2 r_{min})}$ and
 from (\ref{qkncond2}), $q$ must satisfy
\begin{align}\label{qint}
q^*<q<2^{2n/k}q^{*}.
\end{align}
Finally, to show that for each $n>4$ in our sequence a corresponding $q$ exists that satisfies \eqref{qint}, we use the following lemma.
\begin{lemma}[\cite{german}]
There always exists a natural prime congruent to $1\mod 3$ between integers $m$ and $2m$ where $m>4$.
\end{lemma}

We would also like to note that from (\ref{qkncond1}), the growth of $q$ is $O(n^{\frac{1}{\beta}})$. Thus,
\begin{align}\label{dgo01}
 \underset{{n\rightarrow\infty}}{\text{lim}}\hspace{0.2in}n/q=0.
\end{align}

\subsection{Proof: Existence of $\mathbb{Z}[\omega]$-lattices that are good for covering}\label{sec:GfC}
The proof of this theorem is divided into two parts. In the first part, sufficient conditions are obtained such that most Eisenstein lattices in the ensemble are ``almost complete" $\mathcal{V}$-coverings. In the second part, stricter conditions are imposed such that most of the Eisentein lattices in the ensemble are $complete$ $\mathcal{V}$-coverings and thus $complete$ coverings .

\textbf{Part I: Almost complete covering}\\

Denote $d$ to be half of the largest distance between any two points that lie within the Voronoi region of an element in $\text{GRID}$.
\begin{align}\label{dgo02}
d=\sqrt{\frac{n}{3q}}.
\end{align}
Note that by (\ref{qint}), $d\rightarrow 0$ as $n\rightarrow\infty$. 

Consider the lattice constellation $\Lambda^*$ of the ensemble and define $k_1,k_2$ such that $k_1+k_2=k$. We shall denote the Eisenstein lattice constellation obtained from the first $k_1$ columns of $\mathbf{G}$ by $\Lambda^*[k_1]$ and let $\Lambda^*[k_1+j],j=1,\cdots,k_2$ denote the Eisenstein lattice constellation obtained from the first $k_1+j$ columns of $\mathbf{G}$. Let $\underline{x}$ be an arbitrary point such that $\underline{x}\in\mathcal{V}$. Let $\mathcal{S}_1(\underline{x})$ denote the set of GRID points within a modulo distance $r-d$ from $\underline{x}$ where $d$ was defined in (\ref{dgo02}).
\begin{align}
\mathcal{S}_1(\underline{x})=\text{GRID}^*\cap\left(\underline{x}+\mathcal{B}(r-d) \right)^*.
\end{align}
Furthermore, denote $\mathcal{S}_2(\underline{x})$ to be the set of GRID points such that their Voronoi regions intersect a sphere of radius $r-2d$ centered at $\underline{x}$.
\begin{align}
\mathcal{S}_2(\underline{x})=\left\{ \underline{y}\in \text{GRID}^*:\left( \underline{y}+\varrho^{-1}\mathcal{V}\right)\cap\left( \underline{x}+\mathcal{B}(r-2d) \right)^* \right\}.\nonumber\\
\end{align}
It can be observed that $\mathcal{S}_2(\underline{x})\subset\mathcal{S}_1(\underline{x})$. Thus, the cardinality of $\mathcal{S}_1(\underline{x})$ can be bounded as:
\begin{align}
|\mathcal{S}_1(\underline{x})|&\geq |\mathcal{S}_2(\underline{x})|\geq\left\lceil V_{\mathcal{B}}(r-2d)/\text{Vol}(\varrho^{-1}\mathcal{V}) \right\rceil\nonumber\\
&=\left\lceil q^n (\sqrt{3}/2)^{-n}V_{\mathcal{B}}(r-2d)\right\rceil.
\end{align}
By the second property of the ensemble, the probability that $\underline{x}$ is covered by a sphere of radius $(r-d)$ centered at any point of $\Lambda^*{[k_1]}$ satisfies
\begin{align}\label{lbacc}
\text{Pr}\left\{ \underline{x}\in\left( \Lambda_i^*{[k_1]}+\mathcal{B}(r-d)\right)^* \right\}&=\nonumber\\
|\mathcal{S}_1(\underline{x})|/q^n&\geq (\sqrt{3}/2)^{-n}V_{\mathcal{B}}(r-2d),\nonumber\\
\end{align}
for $i=1,\cdots,M_1-1$ where $M_1=q^{k_1}$ and $\Lambda_i^*$ is the $i$th point of $\Lambda^*$. The indicator random variable $\eta_i$ for $i=1,\cdots,M_1-1$ is defined as
\[
 \eta_i=\eta_i(\underline{x})
  \begin{cases}
   1, & \text{if } \underline{x}\in \left( \Lambda_i^*{[k_1]}+\mathcal{B}(r-d)\right)^* \\
   0,       & \text{otherwise }
  \end{cases}
\]
Note that $i=0$ is not considered since $\Lambda_0^*[k_1]=0$ deterministically. Thus, $\eta_i$ is statistically independent of both $i$ and $\underline{x}$. Define $\mathcal{X}=\mathcal{X}(\underline{x})$ as follows:
\begin{align}
\mathcal{X}=\sum_{i=1}^{M_1-1}\eta_i.
\end{align}
Hence, $\mathcal{X}$ is equal to the number of nonzero codewords $(r-d)$-covering $\underline{x}$. Computing the expectation of $\mathcal{X}$ and using the lower bound from (\ref{lbacc}),
\begin{align}\label{lbacc2}
E(\mathcal{X})&=\sum_{i=1}^{M_1-1}E(\eta_i) \nonumber \\
&\geq\left(M_1-1\right)(\sqrt{3}/2)^{-n}V_{\mathcal{B}}(r-2d).
\end{align}
Since the $\eta_i$'s are pairwise independent and thus uncorrelated, similar to \cite{erez2005lattices} one has
\begin{align}\label{cheby}
\text{Var}(\mathcal{X})\leq E(\mathcal{X}).
\end{align}
Using (\ref{cheby}), by Chebyshev's inequality, for any $\nu>0$
\begin{equation}\label{cheby2}
\text{Pr}\left\{ |\mathcal{X}-E(\mathcal{X})|>2^{\nu}\sqrt{E(\mathcal{X})} \right\}<\frac{\text{Var}(\mathcal{X})}{2^{2\nu}E(\mathcal{X})}\leq2^{-2\nu}.
\end{equation}
Define
\begin{equation}
\mu(\nu)=E(\mathcal{X})-2^{\nu}\sqrt{E(\mathcal{X})}.
\end{equation}
Then from (\ref{cheby2}),
\begin{equation}\label{cheby4}
\text{Pr}\{\mathcal{X}<\mu(\nu)\}<2^{-2\nu}.
\end{equation}
If $\mu(\nu)\geq1$, $\text{Pr}\{\mathcal{X}<1\}$ is upper-bounded by $2^{-2\nu}$ as well.

A point $\underline{x}\in\mathcal{V}$ will be referred as \emph{remote} from a discrete set of points $\mathcal{A}$ if it is not $r-d$-covered by $\left(\mathcal{A}+\mathcal{B}(r-d)\right)^*$, i.e. if $\underline{x}$ does not belong to an $(r-d)$- sphere centered at any point of $\mathcal{A}$. Therefore, $\mathcal{X}(\underline{x})<1$ implies that ``$\underline{x}$ is remote from $\Lambda^*[k_1]$''. Define $\mathcal{Q}\left(\mathcal{A} \right)$ to be the set of (continuous) points which are remote from the discrete set $\mathcal{A}$. Denote $\mathcal{Q}_i=\mathcal{Q}\left( \Lambda^*[k_1+i]\right),i=0,1,\cdots,k_2$ and define
\begin{equation}
q_i=|\mathcal{Q}_i|/\text{Vol}\left(\mathcal{V}\right),
\end{equation}
to be the fraction of (continuous) points in $\mathcal{V}$ which are remote from $\Lambda^*[k_1+i]$. Then,
\begin{align}
|\mathcal{Q}_0|&=\int_{\mathcal{V}}{\mathbf{1}\left(\mathcal{X}(\underline{x})<1 \right)d{\underline{x}}}\\
&\leq\int_{\mathcal{V}}{\mathbf{1}\left(\mathcal{X}(\underline{x})< \mu(\nu) \right)d{\underline{x}}},
\end{align}
under the condition that $\mu(\nu)>1$. Then, from (\ref{cheby4}) we have
\begin{equation}\label{cheby3}
E(q_0)<2^{-2\nu}.
\end{equation}
Applying Markov's inequality we get
\begin{equation}
\text{Pr}\{q_0>2^{\nu}E(q_0)\}<2^{-\nu}.
\end{equation}
Using (\ref{cheby3}),
\begin{equation}\label{ineq_q0}
\text{Pr}\{q_0>2^{-\nu}\}<2^{-\nu}.
\end{equation}
Therefore, by taking $\nu\rightarrow\infty$ and keeping $\mu(\nu)\geq1$, this probability can be made arbitrarily small as $n\rightarrow\infty$. In order to satisfy these constraints it is sufficient to take $\nu=o(\log n)$ and $E(\mathcal{X})>n^{\lambda}$ for some $\lambda>0$. By (\ref{lbacc2}) this would be satisfied if we choose a radius $r$ such that
\begin{equation}\label{almcompend}
q^{k_1}-1=\frac{n^{\lambda}}{V_{\mathcal{B}}(r-2d)}\left(\sqrt{3}/2\right)^n.
\end{equation}
Hence, we conclude that for these choice of parameters, for most lattices chosen from the $\left(n,k,q,\mathbb{Z}[\omega]\right)$ ensemble, \emph{almost all} points are covered by spheres of radius $r-d$.\\

\textbf{Part II:  Complete covering}\\

We would like to obtain an ensemble of $\mathbb{Z}[\omega]$-lattices such that most of its members are able to cover all the points in $\mathcal{V}$. $\mathcal{Q}(\mathcal{A})$ is redefined to be the set of $\text{GRID}^*$ points, i.e., $\underline{x}\in\text{GRID}^*$ which are remote from $\mathcal{A}$ and $q_i$ is redefined to be the fraction of $\text{GRID}^*$ points that are remote from $\Lambda^*[k_1+i]$. Therefore, an $(r-d)$-covering of all GRID points implies an $r$-covering of all points in $\mathcal{V}$.

By augmenting the generator matrix $\mathbf{G}$ with an additional small number of columns $k_2(k_2\ll k_1)$, the fraction of uncovered $\text{GRID}^*$ points can be made smaller than $1/|\text{ GRID}^*|$ which implies that all $\text{ GRID}$ points are $r-d$-covered. We proceed as follows.

Choose $k_1$ and $q$ such that $k_1$ grows faster than $\log^2 n$ and \eqref{qkncond1} and \eqref{qkncond2} are satisfied. Define the set
\begin{equation}
\mathcal{S} = \Lambda^*[k_1]\cup \left(\Lambda^*[k_1]+\left\{\sigma^{-1}(\mathbf{G}_{k_1+1})\cap\mathcal{V}\right\}\right),
\end{equation}
where $\sigma$ is the ring isomorphism defined in section~\ref{sec:ConstrAEis}. Also note that,
\begin{equation}
\Lambda^*[k_1+1]=\bigcup_{m=0}^{q-1}\left(\Lambda^{*}[k_1]+\sigma^{-1}\left(\left[m\cdot(\mathbf{G}_{k_1+1})\right]\mod q\right) \right).
\end{equation}
Hence, $\mathcal{S}\subset\Lambda^*[k_1+1]$ and $q_1$ is upper-bounded by $\frac{\mathcal{Q}\left(\mathcal{S}\right)}{|\text{GRID}|^*}$. Since $\Lambda^*[k_1]+\left\{\sigma^{-1}(\mathbf{G}_{k_1+1})\cap\mathcal{V}\right\}$ is an independent shift of $\Lambda^*[k_1]$, conditioned on $\Lambda^*[k_1]$, the event that $\underline{x}$ is remote from $\Lambda^*[k_1]+\left\{\sigma^{-1}(\mathbf{G}_{k_1+1})\cap\mathcal{V}\right\}$ is independent from whether $\underline{x}$ is remote from $\Lambda^*[k_1]$ and the probability of such an event is $q_0$. Then,
\begin{equation}
E \left\{ \frac{|\mathcal{Q}(\mathcal{S})|}{|\text{GRID}^*|}\Big|q_0 \right\}=q_0^2.
\end{equation}
Due to the fact that $\mathcal{S}\subset\Lambda^*[k_1+1]$, we have $E\left\{ q_1|q_0\right\}\leq q_0^2$. By Markov's inequality,
\begin{equation}
\text{Pr}\Big\{ q_1>2^{\gamma}E(q_1|q_0)\Big|q_0\Big\}.
\end{equation}
Therefore,
\begin{equation}
\text{Pr}\Big\{ q_1\leq2^{\gamma-2\nu}\Big|q_0\leq2^{-\nu}\Big\}\geq1-2^{-\gamma}.
\end{equation}
From Bayes' rule and \eqref{ineq_q0},
\begin{align}
\text{Pr}\Big\{ q_1\leq2^{\gamma-2\nu}\Big\}&\geq\text{Pr}\Big\{ q_1<2^{\gamma-2\nu},q_0\leq2^{-\nu}\Big\}\\
&\geq\left(1-2^{-\gamma}\right)\left(1-2^{-\nu}\right).
\end{align}
Repeating this procedure for $l=0,1,\dots,k_2-1$, we obtain
\begin{align}
q_{l+1}&\leq2^{\gamma}E(q_{l+1}|q_l)\\
&\leq2^{\gamma}q_l^2,
\end{align}
with probability at least $1-2^{-\gamma}$. Hence, the intersection of all these $k_2$ events and the event that $q_0<2^{-\nu}$ has the probability $\left(1-2^{-\nu} \right)\left(1-2^{-\gamma} \right)^{k_2}$, which implies
\begin{equation}
q_{k_2}\leq 2^{2^{k_2}(\gamma-\nu)-\gamma}.
\end{equation}
We would like to choose $k_2$ such that
\begin{equation}\label{qk2bound}
q_{k_2}< q^{-n}=2^{-n\log{q}}.
\end{equation}
The interpretation of \eqref{qk2bound} is $q_{k_2}=0$ since there are $q^n$ points in $\text{GRID}^*$. Therefore, choosing $\gamma=\nu-1$ and
\begin{align}
k_2=\lceil\log{n}+\log{\log{q}}\rceil,
\end{align}
or faster suffices. Due to the fact that $k=k_1+k_2$, we conclude that with probability at least
\begin{equation}\label{newvcond}
\left(1-2^{-\nu}\right)\left(1-2^{-\nu+1}\right)^{(\log{n}+\log{\log{q}})}
\end{equation}
$\Lambda^*[k]$ satisfies $q_{k_2}<q^{-n}$, in other words every $\underline{x}\in \text{GRID}^*$ is covered by at least one sphere of radius $(r-d)$. We would like to impose a condition on $\nu$ such that both $\nu\rightarrow\infty$ and the probability in \eqref{newvcond} goes to 1 as $n\rightarrow\infty$. It suffices to choose
\begin{equation}
\nu=2\log\left( \log{n}+\log{\log{q}}\right).
\end{equation}
Note that as $\mu(\nu)\geq1$, the probability that there remains a point $\underline{x}\in\text{GRID}^*$ that is not $(r-d)$-covered is arbitrarily small as $n\rightarrow\infty$. If every point of $\text{GRID}^*$ is $(r-d)$-covered, then $\mathcal{V}$ is $r$-covered. Thus, the probability of a complete covering with spheres of radius $r$ goes to 1 where $r$ satisfies(see \eqref{almcompend})
\begin{align}
M=q^{k_1+k_2}&=\frac{n^\lambda}{V_{\mathcal{B}}(r-2d)}\left(\sqrt{3}/2\right)^n q^{k_2}\label{compend1}\\
&\leq\frac{n^\lambda}{V_{\mathcal{B}}(r-2d)}\left(\sqrt{3}/2\right)^nq^{(\log{n}+\log{\log{q}})+1}\\
&=\frac{n^\lambda}{V_{\mathcal{B}}(r-2d)}\left(\sqrt{3}/2\right)^n 2^{\log{q}[(\log{n}+\log{\log{q}})+1]}.\label{compend2}
\end{align}
From \eqref{compend1} and \eqref{compend2},
\begin{align}
\frac{r}{r_{\Lambda}^{\text{eff}}}&=\sqrt[2n]{\frac{V_{\mathcal{B}}(r)}{V_{\mathcal{B}}(r-2d)}n^{\lambda}q^{k_2}}\label{compend3}\\
&\leq\left(\frac{r}{r-2d}\right)\cdot n^{\lambda/2n}\cdot2^{(\log{q}\log{n}+\log{q}\log{\log{q}}+\log{q})/2n}.\label{compend4}
\end{align}
For $\rho_{\text{cov}}\rightarrow1$, the left-hand side of \eqref{compend3} should go to 1. Hence, we require each of the three terms on the right-hand side of \eqref{compend4} goes to 1. From \eqref{dgo01} and \eqref{dgo02}, it follows that $d\rightarrow0$ as $n\rightarrow\infty$ provided that $k\leq\beta n$ and $\beta<1$. Therefore,
\begin{equation}
\lim_{n\rightarrow\infty}\left(\frac{r}{r-2d}\right)=1.
\end{equation}
For any fixed $\lambda>0$, we have $\lim_{n\rightarrow\infty}n^{\lambda/2n}=1$. Also, since $k$ grows faster than $\log^2n$, by~\eqref{qkncond1} we have $\log{p}$ grows slower than $o\log(n/\log{n})$. Then,
\begin{equation}
\lim_{n\rightarrow\infty}2^{(\log{q}\log{n}+\log{q}\log{\log{q}}+\log{q})/2n}=1.
\end{equation}
Thus, we have that $\frac{r_{\Lambda}^{\text{cov}}}{r_{\Lambda}^{\text{eff}}}\rightarrow1$ in probability as $n\rightarrow\infty$ which completes the proof.

\subsection{Proof: Existence of good nested $\mathbb{Z}[\omega]$-lattices  }\label{sec:GFACC}

Using our result from Theorem~\ref{covering}, let $\Lambda$ be an $n$-dimensional $\mathbb{Z}[\omega]$-lattice
obtained through Construction-A  with a corresponding generator matrix $\mathbf{B}$ which is good for covering.

\begin{definition}
A set $\mathcal{C}$ of linear $(n,k)$ linear code over $\mathbb{F}_q^n$ is \emph{balanced} if every nonzero element of $\mathbb{F}_q^n$ is contained in the same number, denoted by $N_{\mathcal{C}}$ of codes from $\mathcal{C}$.
\end{definition}

Note that for fixed $n,$ $k,$ and $q$, the set of all linear $(n,k)$ codes over $\mathbb{F}_q$ is balanced. We shall now state Lemma 1 in \cite{loeliger}.
\begin{lemma}\label{bal}
Let $f(\cdot)$ be an arbitrary mapping $\mathbb{F}_q^n\rightarrow\mathbb{R}$ and let $\mathcal{C}$ be a balanced set of linear $(n,k)$ codes over $\mathbb{F}_q$. Then, the average over all linear codes $C$ in $\mathcal{C}$ of the sum $\sum_{c\in C^{\prime}} f(c)$ is given by
\begin{align}
    \frac{1}{\mathcal{C}}\sum_{C\in\mathcal{C}}\sum_{c\in C^{\prime}}f(c)=\frac{q^k-1}{q^n-1}\sum_{v\in\left(\mathbb{F}_q^n\right)^{\prime}}f(v).
\end{align}
\end{lemma}

For proving Theorem~\ref{GfACC}, we shall use nested $\mathbb{Z}[\omega]$-lattices obtained from Construction-A as mentioned in Section~\ref{Sec:NestedZw}.
A scaled version of $\Lambda_C$ denoted as $\gamma\Lambda_C$, where $\gamma\in\mathbb{R}^{+}$ and $\Lambda_C$ was defined in section~\ref{sec:ConstrAEis} is constructed.
 Then, we multiply $\gamma\Lambda_C$ with the generator matrix $\mathbf{B}$ and obtain the lattice $\Lambda_f=\gamma\mathbf{B}\Lambda_C$. It can be observed that
$\gamma\varrho\mathbb{Z}[\omega]^n\subset\gamma\varrho\Lambda\subset\Lambda_f$ and there are $q^k$ elements of $\Lambda_f$ that lie within the
fundamental Voronoi region of $\gamma\varrho\Lambda$. Hence, the volume of the fundamental region
of $\Lambda_f$ is
\begin{align}
\text{Vol}\left(\mathcal{V}_{\Lambda_f}\right)=\gamma^{2n}q^{n-k}\left(\frac{\sqrt{3}}{2}\right)^n\text{Vol}\left(\mathcal{V}_{\Lambda}\right).
\end{align}
We can now extend the Minkowski-Hlawka Theorem in \cite{loeliger} to Eisenstein lattices as follows, following similar steps.


\begin{theorem}(\emph{Minkowski-Hlawka Theorem:})
Let $f$ be a Riemann integrable function $\mathbb{R}^{2n}\rightarrow\mathbb{R}$ of bounded support(i.e., $f(v)=0$ (if $\|v\|$ exceeds some bound). Then for any integer $k$ where $0<k<n$, and any fixed $\text{Vol}(\mathcal{V}_{\Lambda_f})$, the approximation
\begin{align}
\frac{1}{\mathcal{C}} \sum_{C\in\mathcal{C}}\hspace{0.03in}\sum_{v\in g(\gamma\mathbf{B}\Lambda_C^{\prime})}f(v)\approx {\text{Vol}(\mathcal{V}_{\Lambda_f})}^{-1}\int_{\mathbb{R}^{2n}}f(v)dv,
\end{align}
where $\mathcal{C}$ is any balanced set of linear $(n,k)$ codes over $\mathbb{F}_q$ and where $g(\cdot):\mathbb{C}^n\rightarrow\mathbb{R}^{2n}$ as
in \eqref{ctormap}, becomes exact in the limit $q\rightarrow\infty$,
 $\gamma\rightarrow0$, $\gamma^{2n}q^{n-k}\left(\frac{\sqrt{3}}{2}\right)^n\text{Vol}\left(\mathcal{V}_{\Lambda}\right)=\text{Vol}\left(\mathcal{V}_{\Lambda_f}\right)$ fixed. Note that these
conditions imply that $\gamma q\rightarrow\infty$.
\end{theorem}
\begin{proof}
\begin{align}
&\frac{1}{|\mathcal{C}|}\sum_{C\in\mathcal{C}}\hspace{0.03in}\sum_{v\in g(\gamma\mathbf{B}\Lambda_C^{\prime})}f(v)\\
&=\frac{1}{|\mathcal{C}|}\sum_{C\in\mathcal{C}} \Big[\sum_{v\in g\left(\left(\mathbb{Z}[\omega]^n\right)^{\prime}\right):\tilde\sigma(v)=0}f(\gamma\mathbf{B}v)\hdots\nonumber\\
&\hdots\hspace{0.08in}+\sum_{v\in g\left(\mathbb{Z}[\omega]^n\right):\tilde\sigma(v)\in C^{\prime}}f(\gamma\mathbf{B} v)\Big]\\
&=\sum_{v\in\left(g\left(\mathbb{Z}[\omega]^n\right)^{\prime}\right):\tilde\sigma(v)=0}f(\gamma \mathbf{B}v)\nonumber\\\label{MH1}
&+\frac{1}{|\mathcal{C}|}\sum_{C\in\mathcal{C}} \sum_{c\in\mathcal{C}^{\prime}}\left[ \sum_{v\in g\left(\mathbb{Z}[\omega]^n\right):\tilde\sigma(v)=c} f(\gamma\mathbf{B}v) \right]\\\label{MH2}
&=\sum_{v\in g\left(\left(\mathbb{Z}[\omega]^n\right)^{\prime}\right):\tilde\sigma(v)=0}f(\gamma\mathbf{B}v)\nonumber\\
&+\frac{q^k-1}{q^n-1}\sum_{c\in(\mathbb{F}_q^n)^{\prime}}\left[ \sum_{v\in g\left(\mathbb{Z}[\omega]^n\right):\tilde\sigma(v)=c} f(\gamma\mathbf{B}v) \right]\\\label{MH3}
&=\sum_{v\in g\left(\left(\mathbb{Z}[\omega]^n\right)^{\prime}\right):\tilde\sigma(v)=0}f(\gamma \mathbf{B}v)\nonumber\\
&+\frac{q^k-1}{q^n-1}\sum_{v\in g\left(\mathbb{Z}[\omega]^n\right):\tilde\sigma(v)\neq0}f(\gamma\mathbf{B}v),\\
\nonumber
\end{align}
where the step from \eqref{MH1} to \eqref{MH2} is due to Lemma \ref{bal} and due to the fact that $f$ has bounded support,
the left term of \eqref{MH3} vanishes for sufficiently large $\gamma q$ and the right term of \eqref{MH3} becomes
\begin{align}\label{MHTres}
\frac{q^k-1}{q^n-1}\sum_{v\in g\left((\mathbb{Z}[\omega]^n)^{\prime}\right)}f(\gamma \mathbf{B}v)&\approx\nonumber\\
&\gamma^{-2n}q^{k-n}\left(\frac{2}{\sqrt{3}}\right)^{n}\text{Vol}(\mathcal{V}_\Lambda)^{-1}\int_{\mathbb{R}^{2n}}f(v)dv,
\end{align}
which becomes exact in the limit as $\gamma\rightarrow0$, $\gamma q\rightarrow\infty$,
 i.e, a Riemann sum approaching to a Riemann integral. Note that the term $\gamma^{-2n}q^{k-n}\left(\frac{2}{\sqrt{3}}\right)^{n}$
appears in front of the integral in \eqref{MHTres} since it is the reciprocal of the volume of the fundamental Voronoi region of $\Lambda_f=\gamma\mathbf{B}\Lambda_C$.
\end{proof}

Suppose now that a transmitter selects a codeword $\underline{x}$ from an Eisenstein lattice $\Lambda\in\mathbb{C}^n$ (or equivalently $\mathbb{R}^{2n}$) and $\underline{x}$ is transmitted over an AWGN channel where a random noise vector $\underline{z}\in\mathbb{C}^n$(or equivalently $\mathbb{R}^{2n}$) gets added with the variance of each $2n$ components equal to ${P_{\underline{z}}}/2$. The receiver obtains $\underline{y}=\underline{x}+\underline{z}$ and tries to recover $\underline{x}$. Furthermore, let $E\subset\mathbb{R}^{2n}$ be a set of typical noise vectors. We say that an \emph{ambiguity} occurs if $\underline{y}$ can be written in more than one way as $\underline{y}=\underline{x}+\underline{e}$ where $\underline{x}\in\Lambda$ and $\underline{e}\in E$. Let $P_{\text{amb}|E}$ be the probability of ambiguity given that $\underline{z}\in E$. Assuming that the receiver is able to recover $\underline{x}$ whenever $\underline{z}\in E$ and there is no ambiguity, the probability of decoding error is upper-bounded by
\begin{align}\label{pamb}
P_e\leq P_{\text{amb}|E}+P(\underline{z}  \notin E).
\end{align}
Due to the fact that Minkowski-Hlawka theorem can be proven for $\Lambda_f$, the following theorem immediately follows.\cite{loeliger}

\begin{theorem}
Let $E$ be a Jordan measurable bounded subset of $\mathbb{R}^{2n}$ and let $k$ be an integer such that $0<k<n$. Then, for any $\delta>0$, for all sufficiently large $q$, and for all sufficiently small $\gamma$, the arithmetic average of $P_{\text{amb}|E}$ over all lattices $\Lambda_f=\gamma\mathbf{B}\Lambda_C$, $C\in\mathcal{C}$, which
we denote as $\overline{P_{\text{amb}|E}}$, is bounded by
\begin{align}
\overline{P_{\text{amb}|E}}<(1+\delta)\text{Vol}(E)/\text{Vol}\left(\mathcal{V}_{\Lambda_f}\right),
\end{align}
where $\mathcal{C}$ is any balanced set of linear $(n,k)$ codes over $\mathbb{F}_q$ and where $\text{Vol}\left(\mathcal{V}_{\Lambda_f}\right)\triangleq\gamma^{2n}q^{n-k}\text{Vol}(\mathcal{V}_\Lambda)\left(\frac{\sqrt{3}}{2}\right)^n$
is the fundamental volume of the lattices $\Lambda_f=\gamma\mathbf{B}\Lambda_C$, $C\in\mathcal{C}$.
\end{theorem}
Note that as $n\rightarrow\infty$, $E$ will approach the shell of a $2n$-dimensional ball with radius $r_{\underline{z}}=\sqrt{nP_{\underline{z}}}$. Thus
\begin{align}
\text{Vol}(E)\leq\text{Vol}(\mathcal{B}(\sqrt{nP_{\underline{z}}}))=\frac{\left(\sqrt{\pi}r_{\underline{z}}^2 \right)^{n}}{\Gamma(n+1)}\hspace{0.2in}\text{as}\hspace{0.1in}n\rightarrow\infty,
\end{align}
which immediately follows that
\begin{align}
\overline{P_{\text{amb}|E}}\leq(1+\delta)\left(\frac{r_{\underline{z}}}{r_{\gamma\mathbf{B}\Lambda_C}^{\text{eff}}}\right)^{2n},
\end{align}
as $n\rightarrow\infty$. This implies that $\overline{P_{\text{amb}|E}}\rightarrow0$ as $n\rightarrow\infty$ for $r_{\underline{z}}<{r_{\gamma\Lambda_C}^{\text{eff}}}$. Hence for a given lattice $\Lambda_f=\gamma\mathbf{B}\Lambda_C$, $P_{\text{amb}|E}\rightarrow0$ in probability as $n\rightarrow\infty$. Taking into account that $P(\underline{z}\notin E)\rightarrow0$ as $n\rightarrow\infty$, from \eqref{pamb} we conclude that $P_e\rightarrow0$ in probability as $n\rightarrow\infty$. This completes the proof.

\bibliography{references}

\begin{thebibliography}{10}

\bibitem{erez2004achieving}
U.~Erez and R.~Zamir, ``Achieving 1/2 log (1+ {SNR}) on the {AWGN} channel with
  lattice encoding and decoding,'' {\em IEEE Trans. Info. Theory}, vol.~50,
  pp.~2293--2314, Oct. 2004.

\bibitem{zamirmultiterminal}
R.~Zamir, S.~Shamai, and U.~Erez, ``Nested linear/lattice codes for structured
  multiterminal binning,'' {\em IEEE Trans. Info. Theory}, vol.~48,
  pp.~1250--1276, Jun. 2002.

\bibitem{Polty}
G.~Poltyrev, ``On coding without restrictions for the {AWGN} channel,'' {\em
  IEEE Trans. Info. Theory}, vol.~40, pp.~409--417, Mar. 1994.

\bibitem{loeliger}
H.~A. Loeliger, ``Averaging bounds for lattices and linear codes,'' {\em IEEE
  Trans. Info. Theory}, vol.~43, pp.~1767--1773, Nov. 1997.

\bibitem{erez2005lattices}
U.~Erez, S.~Litsyn, and R.~Zamir, ``Lattices which are good for (almost)
  everything,'' {\em IEEE Trans. Info. Theory}, vol.~51, pp.~3401--3416, Oct.
  2005.

\bibitem{SC}
O.~Shalvi, N.~Sommer, and M.~Feder, ``Signal codes,'' {\em Info Theory
  Workshop}, pp.~332--336, Mar. 31--Apr. 4 2003.

\bibitem{LDLC}
N.~Sommer, M.~Feder, and O.~Shalvi, ``Low density lattice codes,'' {\em IEEE
  Trans. Info. Theory}, vol.~54, pp.~1561--1585, Apr. 2008.

\bibitem{wilson}
M.~P. Wilson, K.~R. Narayanan, H.~Pfister, and A.~Sprintson, ``Joint physical
  layer coding and network coding for bi-directional relaying,'' {\em IEEE
  Trans. Info. Theory}, vol.~56, pp.~5641--5654, Nov. 2010.

\bibitem{Nam}
W.~Nam, S.~Y. Chung, and Y.~H. Lee, ``Capacity of the gaussian two-way relay
  channel to within 1/2 bit,'' {\em IEEE Trans. Info. Theory}, vol.~56,
  pp.~5488--5494, Nov. 2010.

\bibitem{Niesen}
U.~Niesen and P.~Whiting, ``The degrees of freedom of compute-and-forward,''
  {\em IEEE Trans. Info. Theory}, vol.~58, pp.~5214--5232, Aug. 2012.

\bibitem{nazer2011CF}
B.~Nazer and M.~Gastpar, ``Compute-and-forward: Harnessing interference through
  structured codes,'' {\em IEEE Trans. Info. Theory}, vol.~57, pp.~6463--6486,
  Oct. 2011.

\bibitem{feng}
C.~Feng, D.~Silva, and F.~R. Kschischang, ``An algebraic approach to
  physical-layer network coding,'' {\em IEEE Intl. Symp. on Info. Theory},
  pp.~1017--1021, Jun. 2010.

\bibitem{other_eis}
Q.~T. Sun, J.~Yuan, T.~Huang, and K.~W. Shum, ``Lattice network codes based on
  {Eisenstein} integers,'' {\em IEEE Trans. Comm.}, vol.~61, pp.~2713--2725,
  2013.

\bibitem{conway1999sphere}
J.~H. Conway and N.~J.~A. Sloane, {\em Sphere Packings, Lattices, and Groups}.
\newblock Springer-Verlag, 1999.

\bibitem{Rogers}
C.~A. Rogers, ``A note on coverings,'' {\em Mathematica}, vol.~4, pp.~1--6,
  1957.

\bibitem{Zamir-Feder1996Quan}
R.~Zamir and M.~Feder, ``On lattice quantization noise,'' {\em IEEE Trans.
  Info. Theory}, vol.~42, pp.~1152--1--159, Jul. 1996.

\bibitem{forney}
G.~D. Forney~Jr., ``Coset codes. {I.} {Introduction} and geometrical
  classification,'' {\em IEEE Trans. Info. Theory}, vol.~34, pp.~1123--1151,
  Sep. 1988.

\bibitem{Voronoi_Codes}
J.~H. Conway and N.~J.~A. Sloane, ``A fast encoding method for lattice codes
  and quantizers,'' {\em IEEE Trans. Info. Theory}, vol.~29, pp.~820--824, Nov.
  1983.

\bibitem{Constr_A_Paper}
J.~Leech and N.~J.~A. Sloane, ``Sphere packings and error correcting codes,''
  {\em Canadian Journal of Mathematics}, vol.~23, pp.~718--745, Nov. 1971.

\bibitem{Ordentlich}
O.~Ordentlich and U.~Erez, ``A simple proof for the existence of {``good''}
  pairs of nested lattices,'' {\em arxiv.org}, Sep. 2012.

\bibitem{quarternions}
J.~H. Conway and D.~Smith, {\em On Quaternions and Octonions}.
\newblock CRC Press, 2003.

\bibitem{Algebra_Book}
T.~W. Hungerford, {\em Algebra (Graduate Texts in Mathematics)}.
\newblock Springer, 1974.

\bibitem{german}
R.~Breusch, ``Zur verallgemeinerung des bertrandsehen postulates, da{\ss}
  zwischen x und 2x stets primzalflen liegen,'' {\em Mathematische
  Zeitschrift}, vol.~34, pp.~505--526, 1932.

\bibitem{MIL}
H.~V. Venderson and S.~R. Searle, ``On deriving the inverse sum of matrices,''
  {\em SIAM Review}, pp.~53--60, 1981.

\bibitem{viterbo99}
E.~Viterbo and J.~Boutros, ``A universal lattice code decoder for fading
  channels,'' {\em IEEE Trans. Info. Theory}, vol.~45, pp.~1639--1642, July
  1999.

\bibitem{Agrell}
E.~Agrell, T.~Eriksson, A.~Vardy, and K.~Zeger, ``Closest point search in
  lattices,'' {\em IEEE Trans. Info. Theory}, vol.~48, pp.~2201--2214, Aug.
  2002.

\bibitem{LLL}
A.~K. Lenstra, H.~Lenstra, and L.~Lovasz, ``Factoring polynomials with rational
  coefficients,'' {\em Math Ann.}, pp.~515--534, 1982.

\bibitem{Napias}
H.~Napias, ``A generalization of the {LLL}-algorithm over {Euclidean} rings or
  orders,'' {\em J. Theorie des Nombres de Bordeaux}, pp.~387--396, 1996.

\bibitem{Sakzad}
A.~Sakzad, E.~Viterbo, Y.~Hong, and J.~Boutros, ``On the ergodic computation
  rate for compute-and-forward,'' {\em International Symposium on Network
  Coding (NetCod)}, pp.~131--136, 2012.

\bibitem{viterbo}
A.~Sakzad, J.~Harshan, and E.~Viterbo, ``Integer-forcing {MIMO} linear
  receivers based on lattice reduction,'' {\em IEEE Trans. Wireless Comm.},
  vol.~12, pp.~4905--4915, 2013.

\end{thebibliography}


\begin{thebibliography}{99}

\bibitem{nazer2011CF}
B. Nazer and M. Gastpar, ``Compute-and-Forward: Harnessing Interference through Structured Codes," \emph{IEEE Trans. Info. Theory, } vol. 57, no. 10, pp. 6463-6486, Oct. 2011.

\bibitem{feng}
C. Feng, D. Silva, and F. R. Kschischang, ``An algebraic approach to physical-layer network coding," \emph{IEEE Intl. Symp. on Info. Theory}, pp. 1017-1021, 2010.

\bibitem{erez2005lattices}
U.~Erez, S.~Litsyn, and R.~Zamir,
\newblock {``Lattices which are good for (almost) everything''},
\newblock {\em {IEEE} Trans. Info. Theory}, vol. 51, no.10, pp. 3401--3416, Oct. 2005.

\bibitem{erez2004achieving}
U.~Erez and R.~Zamir,
\newblock {``Achieving 1/2 log (1+ SNR) on the AWGN channel with lattice encoding
  and decoding''},
\newblock {\em {IEEE} Trans. Info. Theory}, vol. 50, no. 10, pp. 2293-2314, Oct. 2004.

\bibitem{zamirmultiterminal}
 R.~Zamir, S.~Shamai, and U.~Erez,
 \newblock {``Nested linear/lattice codes for structured multiterminal binning''},
 \newblock {\em {IEEE} Trans. Info. Theory}, vol. 48, no.6, pp. 1250--1276, Jun. 2002.

 \bibitem{Polty}
G. Poltyrev, ``On coding without restrictions for the AWGN channel'', \emph{IEEE Trans. Info. Theory}, vol. 40, no. 2, pp. 409-417, Mar. 1994.


\bibitem{conway1999sphere}
J.H. Conway and N.J.A. Sloane,
\newblock {\em {Sphere Packings, Lattices, and Groups}}.
\newblock Springer Verlag, 1999.

\bibitem{Zamir-Feder1996Quan}
R. Zamir and M. Feder, ``On lattice quantization noise," \emph{IEEE Trans. Info. Theory}, pp. 1152-1159, vol. 42, no. 4, July 1996.

\bibitem{LDLC}
N.~Sommer, M.Feder, and O. Shalvi, ``Low density lattice codes," \emph{IEEE Trans. Info. Theory}, vol. 54, no. 4, pp. 1561-1585, April 2008.

\bibitem{SC}
O.~Shalvi, N.Sommer, and M. Feder, ``Signal codes," in \emph{ Info. Theory Workshop}, pp. 332--336, 2003.

\bibitem{numbertheory}
B. Fine and G. Rosenberger,
\newblock {\em {Number Theory}}.
\newblock Springer, 2006.

\bibitem{german}
R. Breusch,``Zur verallgemeinerung des Bertrandsehen postulates, daB zwischen x und 2x stets primzalflen liegen'', \emph{Mathematische Zeitschrift}, vol. 34, no.1, pp. 505-526, 1932.

\bibitem{loeliger}
H.~A.~Loeliger, ``Averaging bounds for lattices and linear codes," \emph{IEEE Trans. Info. Theory}, vol. 43, no. 5, pp. 1767-1773, Nov. 1997.


\bibitem{quarternions}
J. H. Conway and D. Smith, \emph{On Quaternions and Octonions}. CRC Press, 2003.

\bibitem{forney}
G. D. Forney Jr., ``Coset codes. I. Introduction and geometrical classification, '' \emph{IEEE Trans. Info. Theory}, vol. 34, no. 5, pp. 1123-1151, Sept. 1998.


\bibitem{Polty}
G. Poltyrev, ``On coding without restrictions for the AWGN channel, '' \emph{IEEE Trans. Info. Theory}, vol. 40, no. 2, pp. 409-417, Mar. 1994.

\bibitem{Rogers}
C. A. Rogers, ``A note on coverings, '' \emph{Mathematica}, vol. 4, pp. 1-6, 1957.

\bibitem{Tomlinson}
M. Tomlinson, `` New Automatic Equalizer Employing Modulo Arithmetic, '' \emph{Elect. Letters}, pp. 138-139, Mar., 1971.

\bibitem{SLDPC}
S. Kudekar, T. Richardson, and R. Urbanke, ``Threshold Saturation via Spatial Coupling: Why Convolutional LDPC Ensembles Perform so well over the BEC, ''\emph{IEEE Trans. Info. Theory}, vol. 57, no. 2, pp. 803-834, Feb. 2011.

\bibitem{Kudekar_universal}
S. Kudekar, T. Richardson, and R. Urbanke, `` Spatially Coupled Ensembles Universally Achieve Capacity under Belief Propagation, ''\emph{arxiv.org}, Jan., 2012

\bibitem{Kudekar}
S. Kudekar, T. Richardson, and R. Urbanke, ``Threshold Saturation via Spatial Coupling: Why Convolutional LDPC Ensembles Perform so well over the BEC, ''\emph{IEEE Trans. Info. Theory}, vol. 57, no. 2, pp. 803-834, Feb. 2011.


\bibitem{wilson}
M. P. Wilson, K. Narayanan, and H. Pfister, and A. Sprintson``Joint Physical Layer Coding and Network Coding for Bi-Directional Relaying, ''\emph{IEEE Trans. Info. Theory}, vol. 56, no. 11, pp. 5641-5654, Nov. 2010.

\bibitem{Nam}
W. Nam, S.Y. Chung, and Y. H. Lee, ``Capacity of the Gaussian  Two-Way Relay Channel to Within 1/2 bit,"  \emph{IEEE Trans. Info. Theory, } vol. 56, no. 11, pp. 5488–5494, Nov. 2010.

\bibitem{Forney}
G. D. Forney Jr., ``Coset Codes. I. Introduction and Geometrical Classification," \emph{IEEE Trans. Info. Theory,} vol. 34, no. 5, pp. 1123-1151, Sep. 1988.

\bibitem{Voronoi Codes}
J. Conway and N. Sloane, ``A Fast Encoding Method for Lattice Codes and Quantizers," \emph{IEEE Trans. Info. Theory,} vol. 29, no. 6, pp. 820-824, Nov. 1983.

\bibitem{Constr A Paper}
J. Leech and N. Sloane, ``Sphere Packings and Error Correcting Codes," \emph{Canadian Journal of Mathematics,} vol. 23, no. 4, pp. 718-745, Nov. 1971.

\bibitem{Algebra Book}
T. W. Hungerford, \emph{Algebra (Graduate Texts in Mathematics)}. Springer, 1974.

\bibitem{Modern Coding Theory}
T. Richardson and R. Urbanke, \emph{Modern Coding Theory}. Cambridge University Press, 2008.

\bibitem{Ordentlich}
O. Ordentlich and U. Erez, ``A Simple Proof for the Existence of ``Good" Pairs of Nested Lattices,"  \emph{arxiv.org}, Sep. 2012.

\bibitem{Niesen}
U. Niesen and P. Whiting, ``The Degrees of Freedom of Compute-and-Forward," \emph{IEEE Trans. Info. Theory}, vol. 58, no. 8, pp. 5214-5232, Aug. 2012.

\end{thebibliography}

\end{document}